\documentclass[aps, prx, showpacs, twoside, twocolumn, longbibliography]{revtex4-2}
\usepackage{times,epsfig,amssymb,amsfonts,amsmath,bm,subfigure,mathtools,amsthm,braket,soul,enumitem,color} 
\usepackage{hyperref}
\usepackage[normalem]{ulem}
\usepackage{graphicx}
\usepackage{float}
\usepackage{natbib}
\newtheorem{theorem}{Theorem}
\newtheorem{lemma}{Lemma}
\newtheorem{sublemma}{Lemma}[lemma]

\newcommand{\stkout}[1]{\ifmmode\text{\sout{\ensuremath{#1}}}\else\sout{#1}\fi}
\definecolor{magenta}{rgb}{1.0, 0.0, 0.56}

\begin{document}

\title{A hybrid-qudit representation of digital RGB images}
\author{ Sreetama Das$^{1,2}$}
\author{Filippo Caruso$^{1,2,3}$}
\affiliation{$^1$Department of Physics and Astronomy, University of Florence, Via Sansone 1, Sesto Fiorentino, I-50019, Italy}
\affiliation{$^2$European Laboratory for Non-Linear Spectroscopy (LENS), University of Florence, Via Nello Carrara 1, Sesto Fiorentino, I-50019, Italy}
\affiliation{$^3$QSTAR and CNR-INO, Largo Enrico Fermi 2 - 50125 Firenze - Italy}


\begin{abstract}
    Quantum image processing is an emerging topic in the field of quantum information and technology. In this paper, we propose a new quantum image representation of RGB images, which is an improvement to all the existing representations in terms of using minimum resource. We use two entangled quantum registers constituting of total 7 qutrits to encode the color channels and their intensities. Additionally, we generalize the existing encoding methods by using both qubits and qutrits to encode the pixel positions of a rectangular image. This hybrid-qudit approach aligns well with the current progress of NISQ devices in incorporating higher dimensional quantum systems than qubits. We then describe the image encoding method using higher-order qubit-qutrit gates, and demonstrate the decomposition of these gates in terms of simpler elementary gates. We use the Google Cirq's quantum simulator to verify the image retrieval. We show that the complexity of the image encoding process is linear in the number of pixels.  Lastly, we discuss the image compression and some basic RGB image processing protocols using our representation.

\end{abstract}

\maketitle

\section{Introduction}
State-of-the-art quantum computation, in both theoretical and experimental grounds, largely uses two-dimensional quantum systems known as `qubits'. The same applies for the present-day noisy intermediate-scale quantum (NISQ) \mbox{computing} processors available for simulating real quantum circuits. Intrinsically, all the quantum systems designed as qubits have more than two energy levels, though the higher levels are maintained as non-functional. It is in \mbox{principle} \mbox{possible} to exploit these levels for quantum computation. Recently there has been a surge of interest to develop theoretical ideas and experimental tools to realize quantum computation using $d>2$ dimensional quantum systems, called `qudits' \cite{MS_2000, loss_2001, bartlett_2002, klimov_2003, ralph_2007, gottesman_1998, luo_2014, chen_2014, DOGRA20143452, gedik_2015, adcock_2016, Lu2019QuantumPE, zeilinger_2020, wang_2020, narvaez_2021, su_2022}. This is usually referred to as `multi-valued quantum computation'. It is \mbox{straightforward} to observe that using qudits increases the storage capacity and reduces the complexity of a quantum circuit, improving the overall performance of a quantum information protocol. Experimental platforms, e.g. photonic systems \cite{zeilinger_2020, Lu2019QuantumPE}, trapped ions \cite{klimov_2003}, continuous spin systems \cite{adcock_2016, bartlett_2002}, nuclear magnetic resonance \cite{DOGRA20143452, gedik_2015}, and molecular magnets \cite{loss_2001} has been used to test qudit-based computation.
The present literature on qudits has shed light on to the \mbox{construction} of elementary gates for quantum circuits where all the qudits have same dimension \cite{MS_2000, Kerntopf_2004, bullock_2005, Di2011, zhong_2013, luo_2014, molmer_2013}. However, quantum \mbox{circuit} composed of qudits of different dimensions, often termed as `hybrid-qudit system', remain relatively unexplored \cite{Daboul_2003, khan_2006, dogra_2015}.

Quantum computation has found its application in a plethora of technological fields, quantum image processing being one of them. It is a rapidly developing topic with potential applications in space science, medicine, automobile engineering etc. Compared to classical bits, one needs logarithmically less number of qubits to store the same amount of classical information. Additionally, quantum properties like superposition and entanglement leads to exponential quantum speed-up of the image processing algorithms \cite{qsobel_2015, suter_2017}. 
The first step of an image processing algorithm is to encode the classical image as a quantum state. A number of different forms of quantum encoding methods have been proposed \cite{FRQI_2011, Sun_2011, NEQR_2013, suter_2017, Sang_2017, OCQR_2018, INCQI_2021}. They differ in their applicability for binary, greyscale or color images, the amount of resources consumed, and the efficiency of the image retrieval process. For binary and greyscale images, Flexible representation of quantum images (FRQI) \cite{FRQI_2011} and Novel-enhanced quantum representation of images (NEQR) \cite{NEQR_2013} are the most commonly used encoding methods in the current literature. Compared to FRQI, NEQR employs higher number of qubits to encode the greyscale pixel values and has a more accurate image retrieval process, which we will discuss in the upcoming sections. For color images, specifically those having an RGB color channel format, several quantum representations inspired from FRQI and NEQR have been proposed. To mention some of them, \cite{Sun_2011} extends the FRQI method to encode the color channel intensities using three qubits, though it suffers the same drawback as the FRQI for image retrieval. On the other hand, the NEQR-based color image encoding methods \cite{Sang_2017, INCQI_2021} use 24 qubits to encode the colors, and have accurate image retrieval. In \cite{OCQR_2018}, the authors show that the number of necessary qubits can be reduced to 10 for encoding the colors, while retaining the accurate image retrieval.


In this work, we aim to reduce the number of quantum units required to encode an RGB image, while adhering to the accurate image retrieval process. 
As it has been shown in \cite{Dong_2022}, following the same approach as NEQR, and replacing the qubits with qutrits, only 6 qutrits are required to encode the greyscale values varying from 0 to 256. This means that an RGB image encoding should take 16 qutrits, which is still pretty large. As we show here, it is possible to encode an RGB image using only 7 qutrits. Using higher than three dimensional qudits can logarithmically reduce the required resources. Interestingly, in \cite{richards2021}, it is argued that systems with three basis states indeed construct the most cost-effective circuit compared to $n\neq 3$ dimensional systems. In the IBMQ platform, one can excite the third energy level of a superconducting qubit and convert it to a qutrit. The Google Cirq quantum computing platform \cite{cirq_developers} provides quantum simulation using qudits and corresponding gates. Rigetti's quantum computation supports the activation, gate implementation and measurement on a real qutrit \cite{rqvm2017}. All the above imply that the qutrit-based quantum computing is imminent, and these motivate us to use qutrits in quantum representation of images.

Our encoding scheme also explores the possibility of using both qubits and qutrits to encode the pixel position and color information respectively, thus giving rise to a hybrid qubit-qutrit circuit. We call this a Hybrid Qudit Quantum Representation (HQDQR) of RGB images. The elementary gates of a hybrid-qudit system has been discussed in \cite{Daboul_2003, khan_2006}. In this manuscript, we show the decomposition of higher-order hybrid qudit gates in terms of elementary single-qudit and hybrid two-qudit gates. This in turn shows that, the number of auxiliary qudits and elementary gates needed for our proposed image encoding is much less compared to the other existing RGB image representations. Thus, HQDQR uses minimum number of quantum units as well as minimum number of gates to encode an RGB image. The complexity of the image encoding process is linear in the number of pixels. 

The paper is organized as follows. In section \ref{sec:qudit_circuit} we revisit some basic gates of qutrit and hybrid qubit-qutrit quantum circuit, and show their decomposition in terms of elementary single-qudit and two-qudit gates. In Sec. \ref{sec:qimage} we briefly discuss the existing RGB image representations. Following that, we introduce our hybrid-qudit representation in Sec. \ref{sec:hybrid_qimage}, calculate the complexity of the image encoding process and discuss the image compression. In Sec. \ref{sec:RGB_operations}, we present some basic RGB image operations using our image representation. Finally, we conclude in Sec. \ref{sec:conclusion}.

section{Quantum circuit using qudits}
\label{sec:qudit_circuit}

The gates used in a quantum circuit apply an unitary transformation on the input state. It has been proved that any unitary transformation on a number of qubits can be asymptotically achieved by repeatedly applying a set of single and two-qubit elementary gates on those qubits, allowing a certain amount of error. This set of elementary gates are called \textit{universal gates} for qubits.
The idea of universality can be extended to qudits. In fact, a number of works have proposed the universal gates in a qudit circuit \cite{MS_2000, Kerntopf_2004, Di2011, luo_2014}.
However, the mathematical idea of universality may not always be suitable to apply for the realization of a complex quantum circuit. In practice, the gates used to achieve a unitary for any number of qudits, should be easy to realize experimentally, so that it can be used in real quantum hardware. There has been a number of experimental proposals to realize basic qudit gates in laboratory.

All the above works assume that the all the qudits in a circuit has same dimension $d$. In principal, it is possible to build a circuit with different dimensional qudits. A handful of works discuss the basic gates in a hybrid qudit circuit, and quantum computation using such systems \cite{Daboul_2003,khan_2006}.

In the following subsections, we introduce the elementary gates of a qutrit circuit and a hybrid qubit-qutrit circuit, which will be relevant for our work discussed in this paper.

\subsection{Quantum computation in qutrit systems}
The Hilbert space of a three-level quantum system or ``qutrit'' is spanned by the orthogonal basis vectors $\{|0\rangle, |1\rangle, |2\rangle\}$. 
The elementary and universal qutrit gates, and the possibility of their physical realization, have been investigated \cite{MS_2000, khan_2003, Kerntopf_2004, Di2011, yurtalan_2020, Goss_2022, luo_2022}. 

\subsubsection{Single qutrit gates}

\textit{Ternary bit-flip gates:} 
In analogy to the qubit bit-flip or $X$ gate, the ternary $X$ gate flips the basis states of a qutrit. 

There are six ternary $X$ gates as listed below \cite{khan_2003, Dong_2022},
\begin{eqnarray}
&\sigma^{x}_{+0}=\begin{pmatrix}
1 & 0 & 0\\
0 & 1& 0\\
0 & 0 & 1
\end{pmatrix}, \hspace{0.3cm}
\sigma^{x}_{+1}=\begin{pmatrix}
0 & 0 & 1\\
1 & 0& 0\\
0 & 1 & 0
\end{pmatrix},\\ \nonumber
&\sigma^{x}_{+2}=\begin{pmatrix}
0 & 1 & 0\\
0 & 0& 1\\
1 & 0 & 0
\end{pmatrix}, \hspace{0.3cm}
\sigma^{x}_{01} = \begin{pmatrix}
0 & 1 & 0\\
1 & 0& 0\\
0 & 0 & 1
\end{pmatrix}, \\ \nonumber
&\sigma^{x}_{12} = \begin{pmatrix}
1 & 0 & 0\\
0 & 0& 1\\
0 & 1 & 0
\end{pmatrix}, \hspace{0.3cm}
\sigma^{x}_{02} = \begin{pmatrix}
0 & 0 & 1\\
0 & 1& 0\\
1 & 0 & 0
\end{pmatrix}.
\label{ternary_bitflip}
\end{eqnarray}

The first one is identity matrix which of course does not change anything. The operators $\{\sigma^{x}_{+1}, \sigma^{x}_{+2}\}$ transform the basis $|x\rangle$ by $|x\rangle\rightarrow |(x+1)\mod 3\rangle$ and $|x\rangle\rightarrow |(x+2)\mod 3\rangle$ respectively, and work on all the basis states simultaneously. Lastly, $\{\sigma^{x}_{01}, \sigma^{x}_{12}, \sigma^{x}_{02}\}$ swaps the two basis states $\{|i\rangle,|j\rangle\}$ in the subscript of $\sigma^{x}_{ij}$, while leaving the third basis state unchanged. Following earlier works, we will use a simple notation for the above gates when using them in our circuit diagrams. The abbreviated form is shown in Fig. \ref{hybrid_circuit}(a).

\textit{Ternary Hadamard gate:} The Hadamard gate $H_{2}$ in a two-dimensional Hilbert space $\mathcal{H}_{2}$ is the quantum Fourier transform from the computational basis to the eigenbasis of Pauli matrix $\sigma^{x}$. In an analogous way, the ternary Hadamard gate $H_{3}$ in the Hilbert space $\mathcal{H}_{3}$ can be defined as the following,
\begin{equation}
    H_{3} = \frac{1}{\sqrt{3}}\begin{pmatrix}
    1 & 1 & 1\\
    1 & e^{i\frac{2 \pi}{3}} & e^{-i\frac{2 \pi}{3}}\\
    1 & e^{-i\frac{2 \pi}{3}} & e^{i\frac{2 \pi}{3}}
    \end{pmatrix}.
    \label{ternary_Hadamard}
\end{equation}
An experimental realization of this gate has been possible \cite{yurtalan_2020} using superconducting qutrits.

\subsubsection{Two-qutrit gates}

\textit{Ternary controlled $X$ gates:} 
A binary controlled $X$ gate is a two-qubit gate such that the if the first qubit is in state $|1\rangle$, the second qubit undergoes the bit-flip operation $X$. If the first qubit is in state $|0\rangle$, nothing changes. The above is the standard notion, though a controlled $X$ gate can be configured such that the target qubit is flipped only when the control qubit is in state $|0\rangle$. In a qutrit system, there are 18 different generalized controlled $X$ operations (three possible control states and six target flips for each of them) possible. Each of them is a $9\times 9$ unitary matrix. For example, if we want to flip the target qutrit state from $|0\rangle$ to $|1\rangle$ when the control qutrit state is $|2\rangle$, the corresponding unitary will be,
\begin{eqnarray}
    &U = (|0\rangle\langle 0| + |1\rangle \langle 1|)\otimes \mathbb{I}_{3} + |2\rangle \langle 2| \otimes \sigma^{x}_{01}\\ \nonumber
    &=\begin{pmatrix}
    1 & 0 & 0 & 0 & 0 & 0 & 0 & 0 & 0\\
    0 & 1 & 0 & 0 & 0 & 0 & 0 & 0 & 0\\
    0 & 0 & 1 & 0 & 0 & 0 & 0 & 0 & 0\\
    0 & 0 & 0 & 1 & 0 & 0 & 0 & 0 & 0\\
    0 & 0 & 0 & 0 & 1 & 0 & 0 & 0 & 0\\
    0 & 0 & 0 & 0 & 0 & 1 & 0 & 0 & 0\\
    0 & 0 & 0 & 0 & 0 & 0 & 0 & 1 & 0\\
    0 & 0 & 0 & 0 & 0 & 0 & 1 & 0 & 0\\
    0 & 0 & 0 & 0 & 0 & 0 & 0 & 0 & 1\\
    \end{pmatrix},
    \label{ternary_controlledX}
\end{eqnarray}
where $\mathbb{I}_{d}$ is a $d$-dimensional identity operator.
The rest 17 unitaries can be constructed in a similar way.

\subsection{Quantum computation in hybrid qubit-qutrit systems}
In a hybrid quantum system constituted of both qubits and qutrits, the single qudit gates remain unchanged. However, there can exist two or multi-qudit gates which act on a system of qubits and qutrits. 

\subsubsection{Two-qudit hybrid gates}
\textit{Hybrid controlled $X$ gate:} For this class of gates, the control can be a qubit and the target can be a qutrit, or the vice versa. Suppose, if the control qubit is in state $|1\rangle$, $\sigma^{x}_{12}$ is applied on the target qutrit. The corresponding unitary is,
\begin{eqnarray}
&U=|0\rangle\langle 0| \otimes \mathbb{I}_{3} + |1\rangle\langle1|\otimes \sigma^{x}_{12} \nonumber \\
&=\begin{pmatrix}
1 & 0 & 0 & 0 & 0 & 0\\
0 & 1 & 0 & 0 & 0 & 0\\
0 & 0 & 1 & 0 & 0 & 0\\
0 & 0 & 0 & 1 & 0 & 0\\
0 & 0 & 0 & 0 & 0 & 1\\
0 & 0 & 0 & 0 & 1 & 0\\
\end{pmatrix}.
\label{qubit-C-qutrit-T}
\end{eqnarray}
On the other hand, if the control is on the qutrit and the target is a qubit, and the qubit state flips when the control qutrit is in state $|2\rangle$, then the unitary can be expressed as,
\begin{eqnarray}
&U=(|0\rangle\langle 0| + |1\rangle\langle 1|)\otimes \mathbb{I}_{2} + |2\rangle\langle 2|\otimes \sigma^{x} \nonumber \\
&=\begin{pmatrix}
1 & 0 & 0 & 0 & 0 & 0\\
0 & 1 & 0 & 0 & 0 & 0\\
0 & 0 & 1 & 0 & 0 & 0\\
0 & 0 & 0 & 1 & 0 & 0\\
0 & 0 & 0 & 0 & 0 & 1\\
0 & 0 & 0 & 0 & 1 & 0\\
\end{pmatrix},
\end{eqnarray}
where $\sigma^{x}$ is the qubit $X$ gate. This unitary is same as that in Eq. \ref{qubit-C-qutrit-T}, but as we will see in the upcoming subsections, this may not always be the case.

\begin{figure}[t]
    \centering
    \subfigure[]{\includegraphics[width=0.39\textwidth]{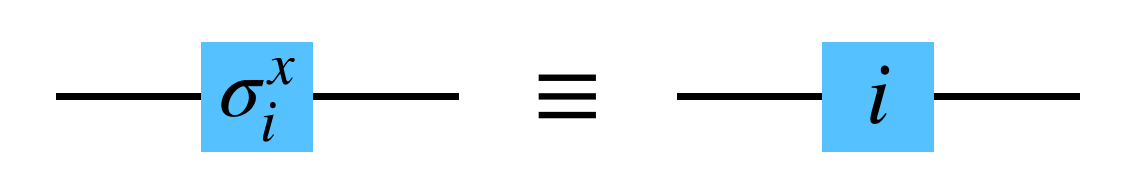}}
    \subfigure[]{\includegraphics[width=0.39\textwidth]{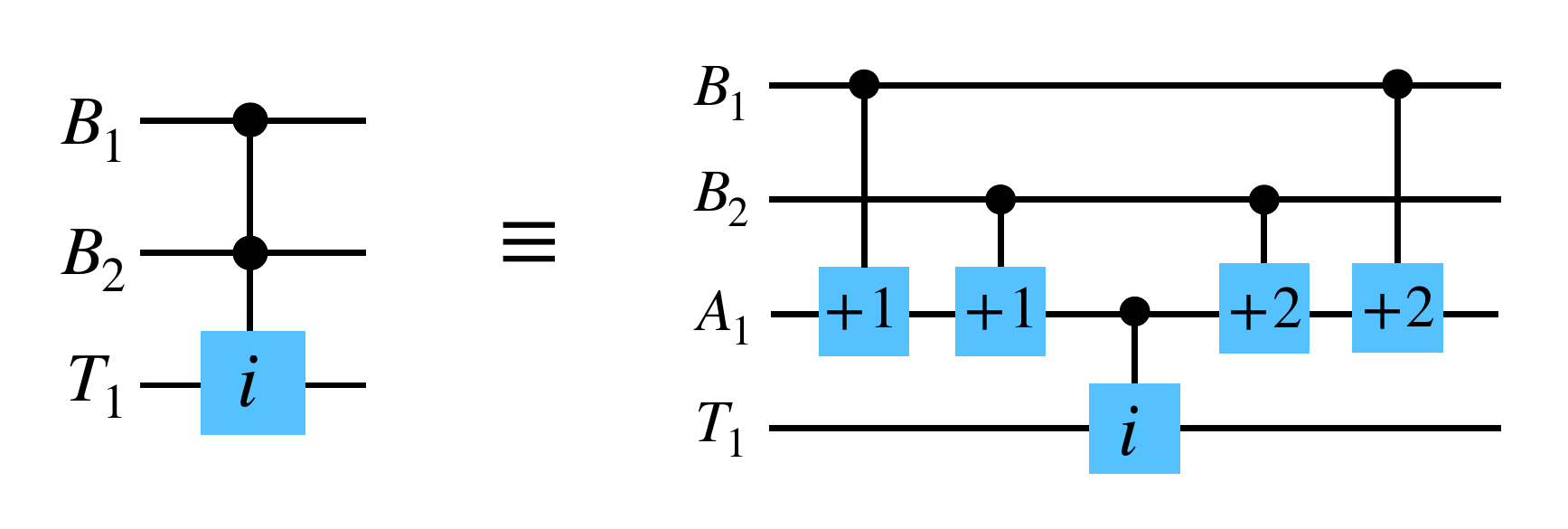}}
    \subfigure[]{\includegraphics[width=0.45\textwidth]{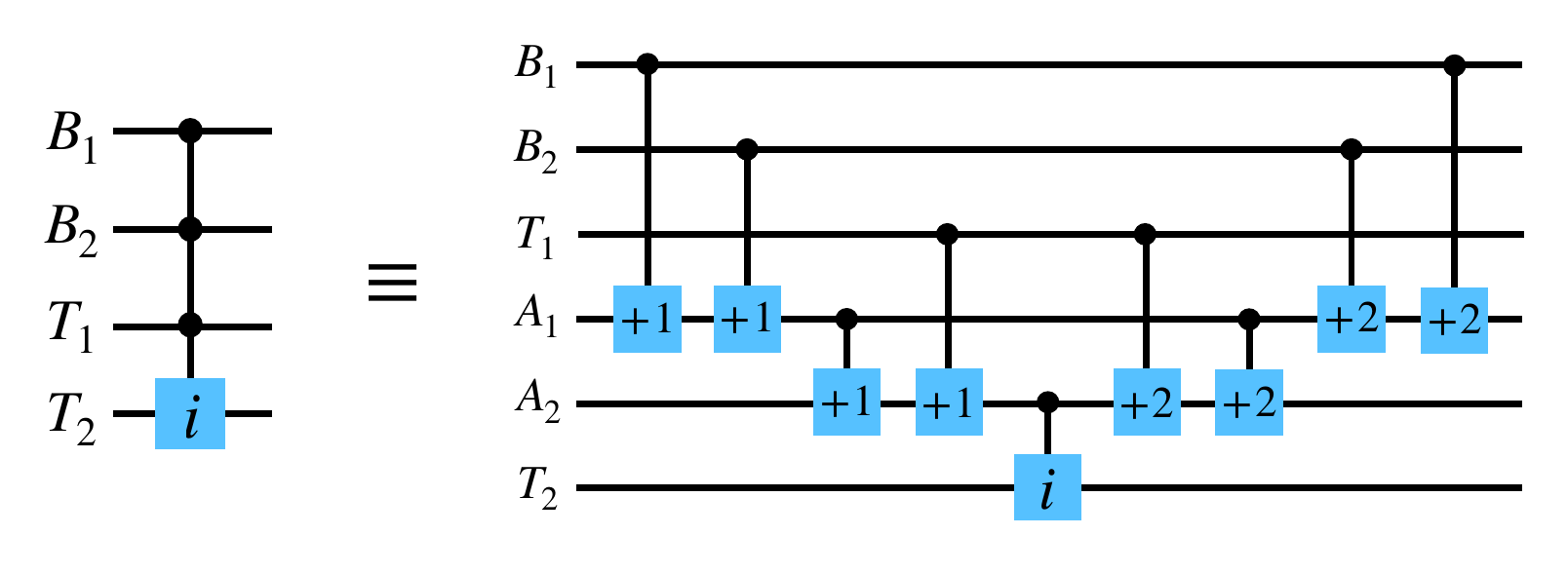}}
    \caption{(a) The notation we use for different controlled X gates in our circuit diagrams, where $i \in \{+0, +1, +2, 01, 12, 02\}$. (b) Decomposition of a three-qudit hybrid Toffoli gate. (Left) A hybrid Toffoli gate where the control is on two qubits $\{B_{1}, B_{2}\}$ and the target is a qutrit $T_{1}$. (Right) The decomposition of the gate using an auxiliary qutrit $A_{1}$ and hybrid controlled $X$ gates. (c) Decomposition of a four-qudit hybrid Toffoli gate. (Left) A hybrid Toffoli gate where the control is on $\{B_{1}, B_{2}\}$ and $T_{1}$, and the target qutrit is $T_{2}$. (Right) The decomposition of the gate using two auxiliary qutrits $\{A_{1}, A_{2}\}$ and hybrid controlled $X$ gates. In both figures, $``i"$ denotes an arbitrary qutrit $X$ operation.}
    \label{hybrid_circuit}
\end{figure}

\subsubsection{Multi-qudit hybrid gates:}

\textit{Hybrid Toffoli gates:}
Now suppose that there are one control qubit and one control qutrit, while the target is a qutrit. If the control qudits are respectively in state $|1\rangle$ and $|2\rangle$, the target qutrit undergoes the bit-flip operation $\sigma^{x}_{12}$. This is analogous to the Toffoli gate or controlled controlled X gate for qubits, and we will call it \textit{hybrid Toffoli gate} in this paper. 
It can be decomposed using simpler single-qudit and hybrid two-qudit gates. We show such a decomposition in Fig. \ref{hybrid_circuit}(b), where there are two control qubits and a target qutrit. The decomposition uses one auxiliary qutrit. If there are $n$ control qudits, we will need $(n-1)$ auxiliary qutrits. The last two hybrid gates are applied on the auxiliary qutrit to bring it back in the initial state $|0\rangle$, so that it can be used for the decomposition of another gate. In Fig. \ref{hybrid_circuit}(c), we show the decomposition of a higher order hybrid Toffoli gate in which the control is on two qubits and a qutrit, and the target is a qutrit.

\begin{theorem}
The complexity of a hybrid Toffoli gate with $n$ number of control qubits and qutrits is $4n-3$.
\end{theorem}
\begin{proof}
There are $n$ controlled $X$ gates from $n$ controls qudits to the auxiliary qutrits, and total $n-2$ controlled $X$ gates among all the adjacent pairs of auxiliary qutrits. The same number of gates are applied to bring back the auxiliary qutrits in state $|0\rangle$. There is one controlled $X$ gate between the last auxiliary qutrit and the target qudit. So, the total number of gates used is $2(n+(n-2))+1=4n-3$.
\end{proof}

\textit{Generalized hybrid Toffoli gates:} For the case discussed above, we assumed that for a control qubit and a control qutrit, the Toffoli gate is activated when the control states are respectively $|1\rangle$ and $|2\rangle$, i.e. when the control qudits are in their highest state. In principal, the gates can be designed to be activated for any of the control states $\{|0\rangle, |1\rangle,..., |d-1\rangle\}$ of a qudit. We call such gates \textit{generalized hybrid Toffoli gate}. The decomposition of a generalized hybrid Toffoli gate is shown in Fig. \ref{general_control}, where $2k$ single qudit gates, $k$ being the number of controls with generalized bit-values, and a higher order hybrid Toffoli gate have been used. The later can be further decomposed in terms of the hybrid two-qudit gates using auxiliary qutrits as shown in Fig. \ref{hybrid_circuit}(c). 
\begin{sublemma}
The complexity of a generalized higher order hybrid Toffoli gate with $n$ control qubits and qutrits is $4n-3+2k$. The maximum value is $6n-3$, reached when $k=n$.
\end{sublemma}

\begin{figure}
    \centering
    \includegraphics[width=0.49\textwidth]{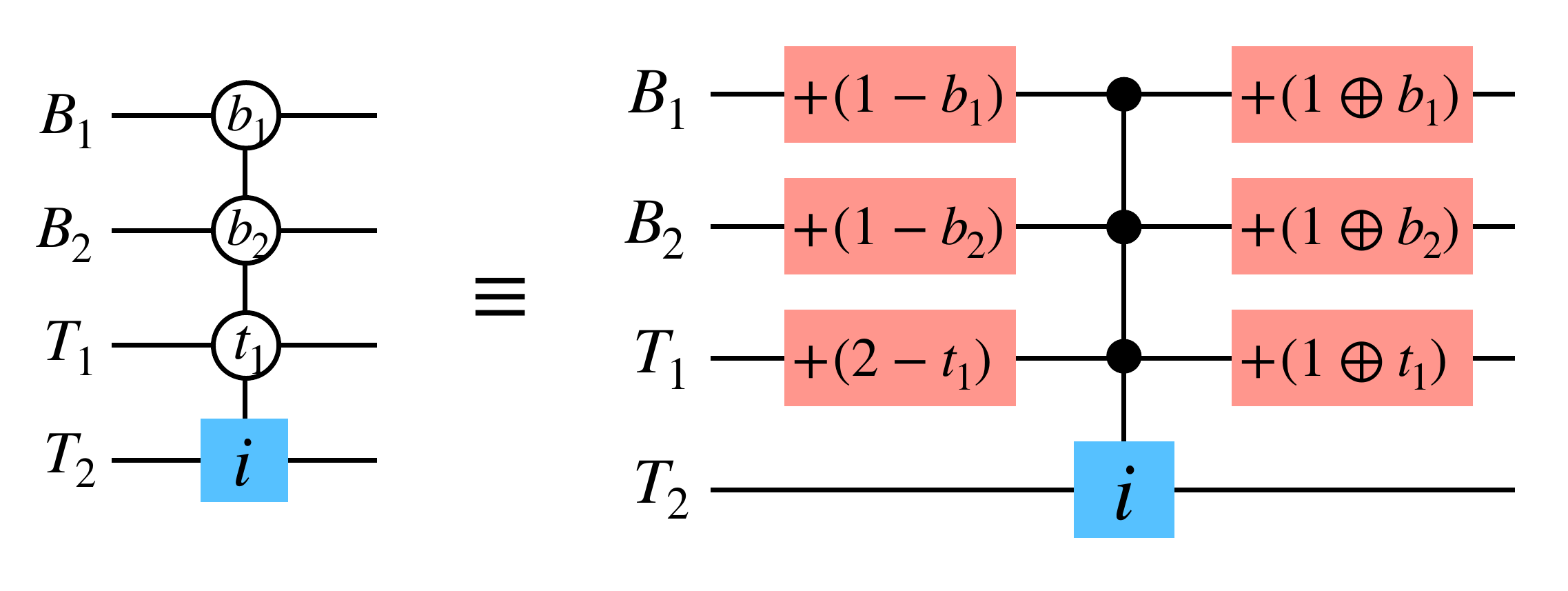}
    \caption{Decomposition of a higher order generalized Toffoli gate. (Left) The control qudits are two qubits $\{B_{1}, B_{2}\}$ and a qutrit $T_{1}$, and the target is a qutrit $T_{2}$. The target qutrit undergoes $\sigma_{i}^{x}$ when the control qudits are in states $|b_{1}\rangle$, $|b_{2}\rangle$ and $|t_{1}\rangle$ respectively. (Right) The decomposition of this gate in terms of $\sigma^{x}_{i}$s and four-qudit hybrid Toffoli gate.}
    \label{general_control}
\end{figure}

\begin{figure}[t]
    \centering
    \subfigure[]{\includegraphics[width=0.39\textwidth]{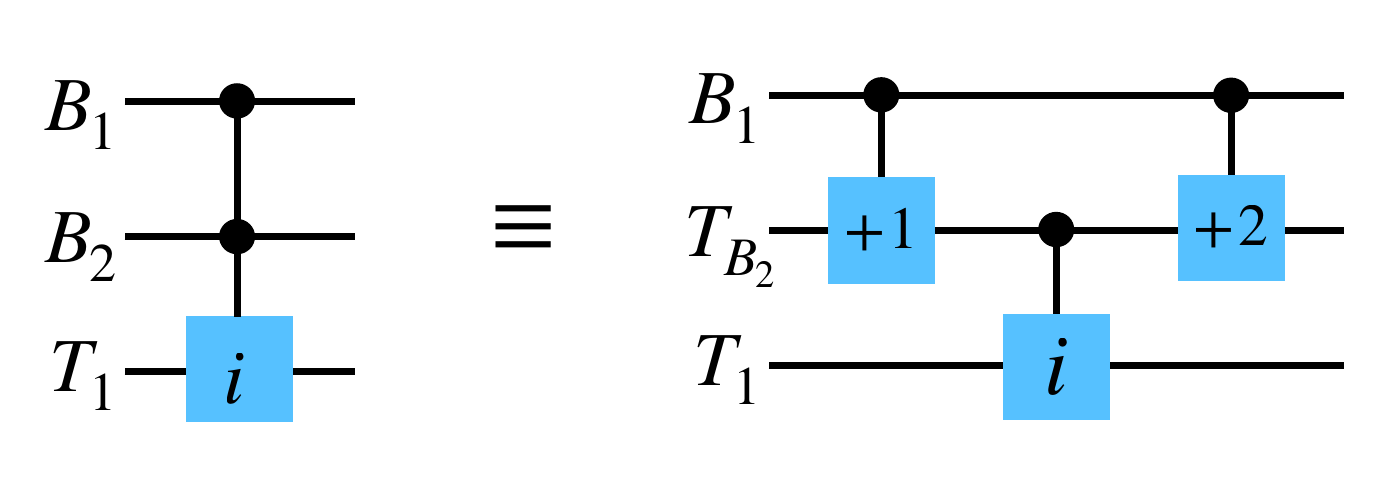}}
    \subfigure[]{\includegraphics[width=0.45\textwidth]{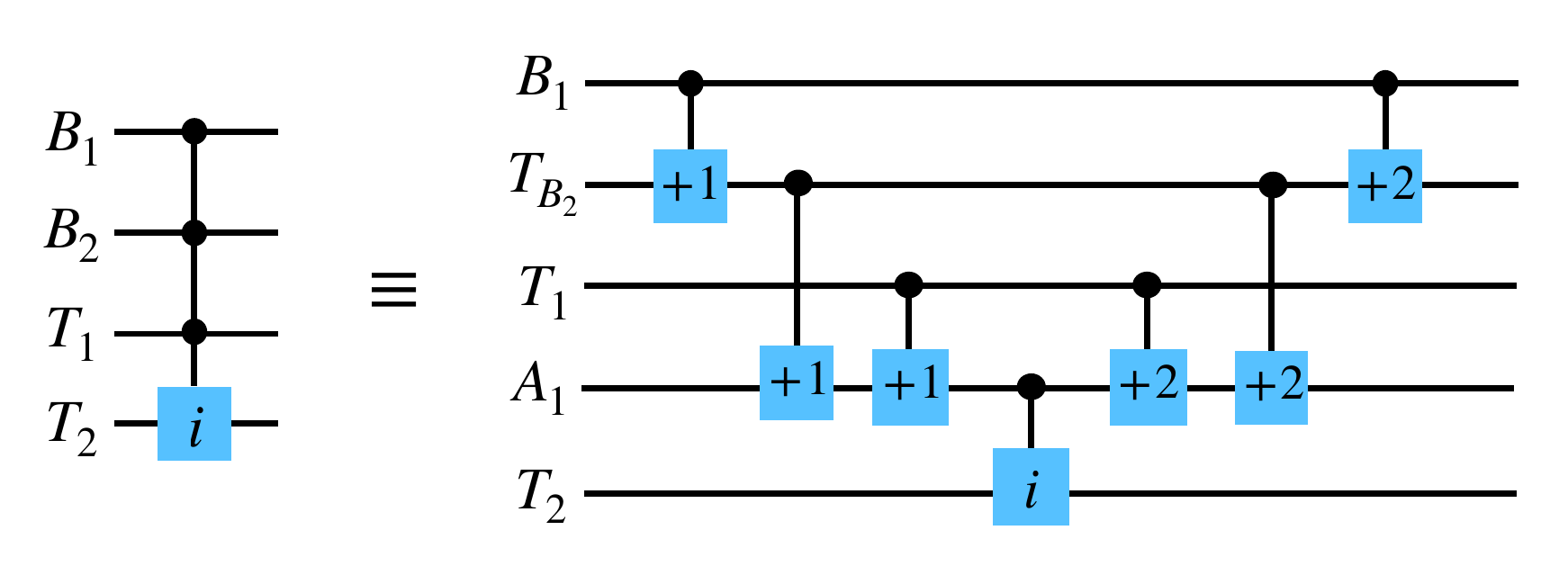}}
    \caption{The improved gate decomposition using effective qutrits. (a) The decomposition of the circuit in Fig. \ref{hybrid_circuit}(b) where the qubit $B_{2}$ becomes effective qutrit $T_{B_{2}}$. (b) The decomposition of the circuit in Fig. \ref{hybrid_circuit}(c), where now one auxiliary qubit $A_{1}$ is required.}
    \label{improved_hybrid_circuit}
\end{figure}

In \cite{gokhale_2019}, the authors demonstrate that the higher order Toffoli gates for qubits can be synthesized without any auxiliary quantum system, if the third energy levels of these qubits are activated. This extra energy level is used only in the intermediate steps to store information, while the input and the output of the circuit remains in the Hilbert space of a qubit. The authors also show that such a circuit has depth $\log N$ for synthesizing a higher order Toffoli gate acting on $N$ qubits. Clearly, the decomposition of a hybrid qubit-qutrit gate can facilitate from a lower circuit depth and less number of elementary gates by using effective qutrits in place of qubits. Such a decomposition is shown in Fig. \ref{improved_hybrid_circuit}.

\begin{theorem}
The complexity of a hybrid Toffoli gate with $n_{1}$ control qubits ( or, $n_1 -1$ effective qutrits ) and $n_{2}$ control qutrits is $2n+2n_{2}-1$, where $n=n_1 + n_2$.
\end{theorem}
\begin{proof}
As evident from Fig. \ref{improved_hybrid_circuit}(b), the number of auxiliary qutrits needed in this case is $n_{2}$. There are $n_{1}-1$ controlled $X$ gates between $n_{1}$ control qubits, and 2 controlled $X$ gates for each of the $n_{2}$ qutrit-auxiliary qutrit pair. Thus, the total number of gates is $2(n_{1}-1+2n_{2})+1=2n+2n_{2}-1$.
\end{proof}
The improvement in the number of gates by using effective qutrits is $4n-3-(2n+2n_{2}-1)=2n_{1}-2$.

\begin{sublemma}
The complexity of a generalized higher order hybrid Toffoli gate with $n_{1}$ control qubits ($n_1 -1$ effective qutrits) and $n_2$ control qutrits is $2n+2n_{2}-2+2k$. 
\end{sublemma}

\section{Quantum image representation}
\label{sec:qimage}

Equipped with the necessary tools to design hybrid qubit-qutrit circuit, in this section, we describe our proposed RGB image representation. In the beginning, we briefly discuss some of the existing RGB image representations, to make a clear comparison of our work with the existing approaches. 

\subsection{Multi-channel representation for images on quantum computers (MCQI)}
In this representation \cite{Sun_2011}, a $2^{n}\times 2^{n}$ dimensional RGB image is encoded using $2n+3$ qubits. The first 3 qubits encode the intensities of the three channels $\{\mathrm{R, G, B}\}$ and the rest $2n$ qubits encode the positions. The intensities of the three channels are encoded using the angles $\theta_{i}$ ($i=R,G,B$). The quantum state corresponding to the image is the following,
\begin{eqnarray}
    |I\rangle =& \frac{1}{2^{n}}\sum\limits_{Y=0}^{2^{n}-1}\sum\limits_{X=0}^{2^{n}-1} (\cos\theta^{i}_{R}|000\rangle + \cos\theta^{i}_{G}|001\rangle + \cos\theta^{i}_{B}|010\rangle \nonumber \\
    &+ \sin\theta^{i}_{R}|100\rangle + \cos\theta^{i}_{G}|101\rangle + \cos\theta^{i}_{B}|110\rangle \nonumber \\
    &+\cos 0|011\rangle + \sin 0|111\rangle)\times |YX\rangle,
\end{eqnarray}
where $|X\rangle$ and $|Y\rangle$ are the $2^{n}$ basis states of an $n$-qubit system. The 3-qubit register has 8 basis states, from which six are used to encode the color information, and the coefficients corresponding to $|011\rangle$ and $|111\rangle$ are set as constants $\sin 0$ and $\cos 0$, so that they do not carry any information. To retrieve the images, one has to perform repeated measurements on all the three color qubits to probabilistically obtain the coefficients $\cos\theta_{i}$ and $\sin\theta_{i}$. Because of the probabilistic nature of the output, such an image retrieval is not accurate, and large number of measurements are required to reach a particular accuracy.

\subsection{Novel quantum representation of color digital images (NCQI)}
This encoding method \cite{Sang_2017} is inspired from NEQR, where the intensities of each channel is encoded using the basis vectors of 8 qubits. So, in total there are 24 qubits needed to encode the color of a pixel. For a $2^{n}\times 2^{n}$ dimensional image, the positions are encoded using $2n$ qubits. So, the total number of qubits in $2n + 24$. The quantum state corresponding to a $2\times 2$ image looks like the following,
\begin{equation}
|I\rangle=\frac{1}{2^{n}} \sum\limits_{Y=0}^{2^{n}-1}\sum\limits_{X=0}^{2^{n}-1}|R_{XY}\rangle|G_{XY}\rangle|B_{XY}\rangle|YX\rangle,
\end{equation}
where $|C_{XY}\rangle=|C^{q-1}_{XY}...C^{1}_{XY}C^{0}_{XY}\rangle$, $C=\{R, G, B\}$ and $q=8$. A projective measurement on the position and the color qubits can accurately retrieve the pixel positions and corresponding pixel colors and intensities. However, one should note that while using the real quantum processors for retrieving the images, one still has to perform a finite number of repeated measurements to obtain all the basis vectors corresponding to the intensities.

The time complexity of preparing quantum images using NCQI is quadratically less compared to that using MCQI. This advantage is similar to the advantage of using NEQR over FRQI. Also, NCQI allows for more complex color transformations, and solves the problem of probabilistic retrieval of pixel values. A limitation of both the above representations is that they consider only square images. To generalize this for rectangular images, an improved encoding method has been proposed \cite{INCQI_2021}.

\subsection{Optimized quantum representation of color images}
In this encoding method, the color information is stored using two entangled quantum registers. The first register with two qubits encodes the channel information, while the second register with eight qubits encodes the intensity of that channel. The pixel positions are again encoded using $2n$ qubits for a $2^{n}\times 2^{n}$ classical image. The quantum image state is,
\begin{eqnarray}
&|I\rangle = \frac{1}{2^{n}}\Big( |R_{XY}\rangle|00\rangle +|G_{XY}\rangle|01\rangle + |B_{XY}\rangle|10\rangle \nonumber \\ &+ |S_{XY}\rangle|11\rangle \Big) \otimes |YX\rangle,
\end{eqnarray}
where $|C_{XY}\rangle$ $(C=R, G, B)$ is the same as defined for NCQI. This representation uses only 10 qubits to encode the colors, which is a significant improvement over NCQI.

\begin{figure*}
    \centering
    \includegraphics[width=0.95\textwidth]{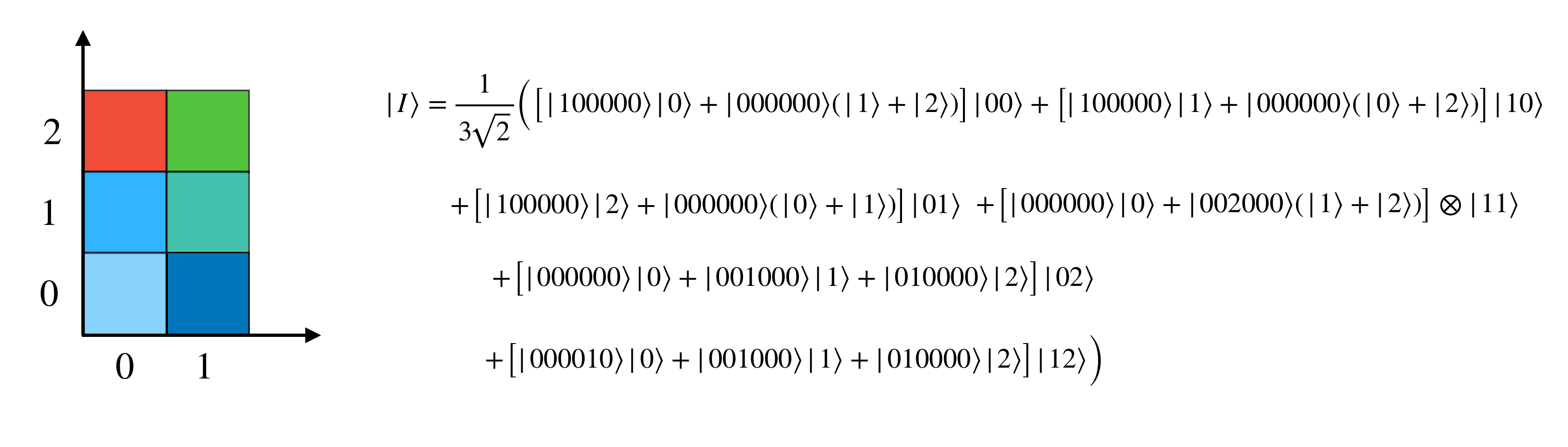}
    \caption{A $3\times 2$ dimensional RGB image and its corresponding quantum image state $|I\rangle$ represented using HQDQR. }
    \label{color_image}
\end{figure*}

\begin{figure*}
    \centering
    \subfigure[]{\includegraphics[width=0.70\textwidth]{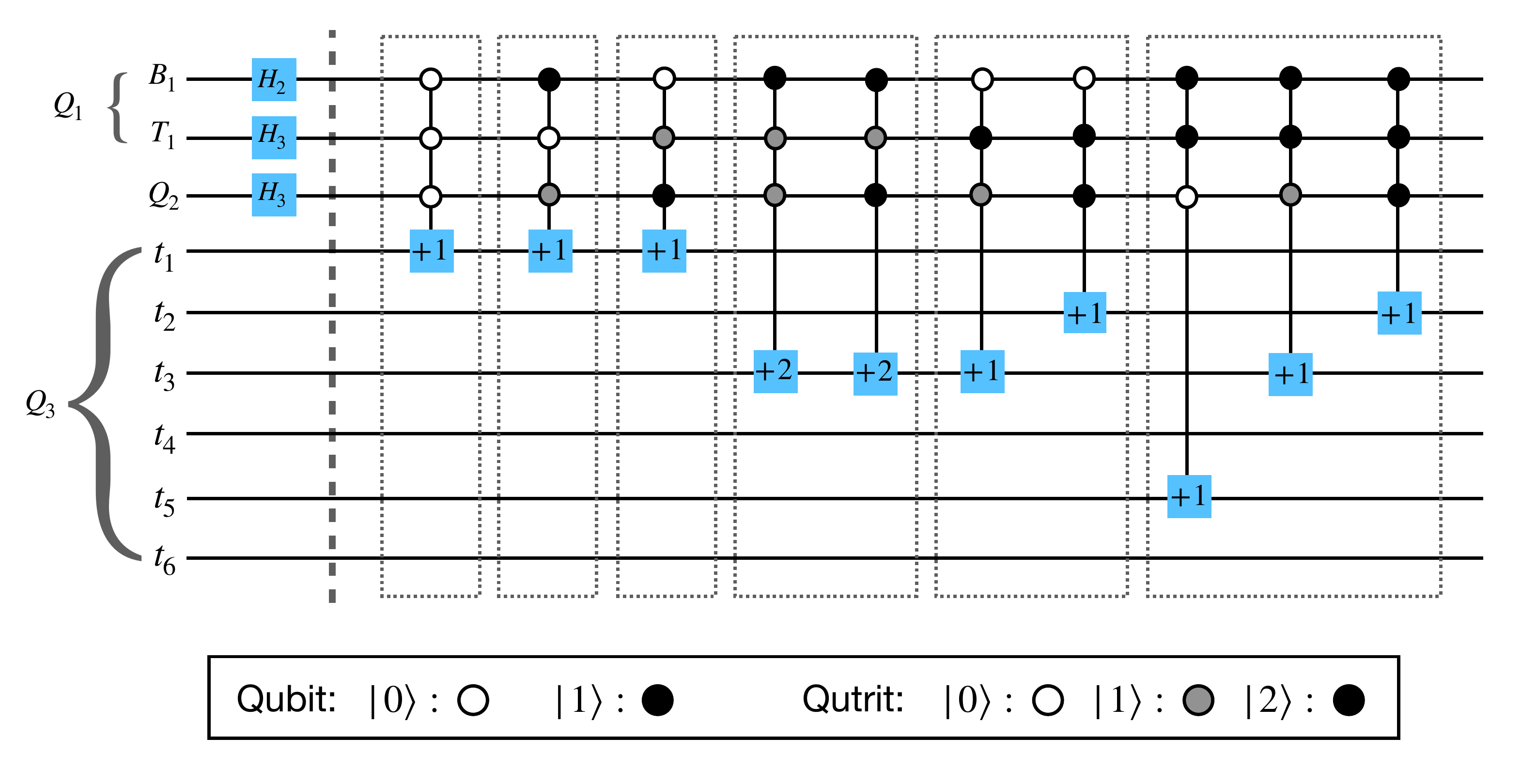}}
    \subfigure[]{\includegraphics[width=0.8\textwidth]{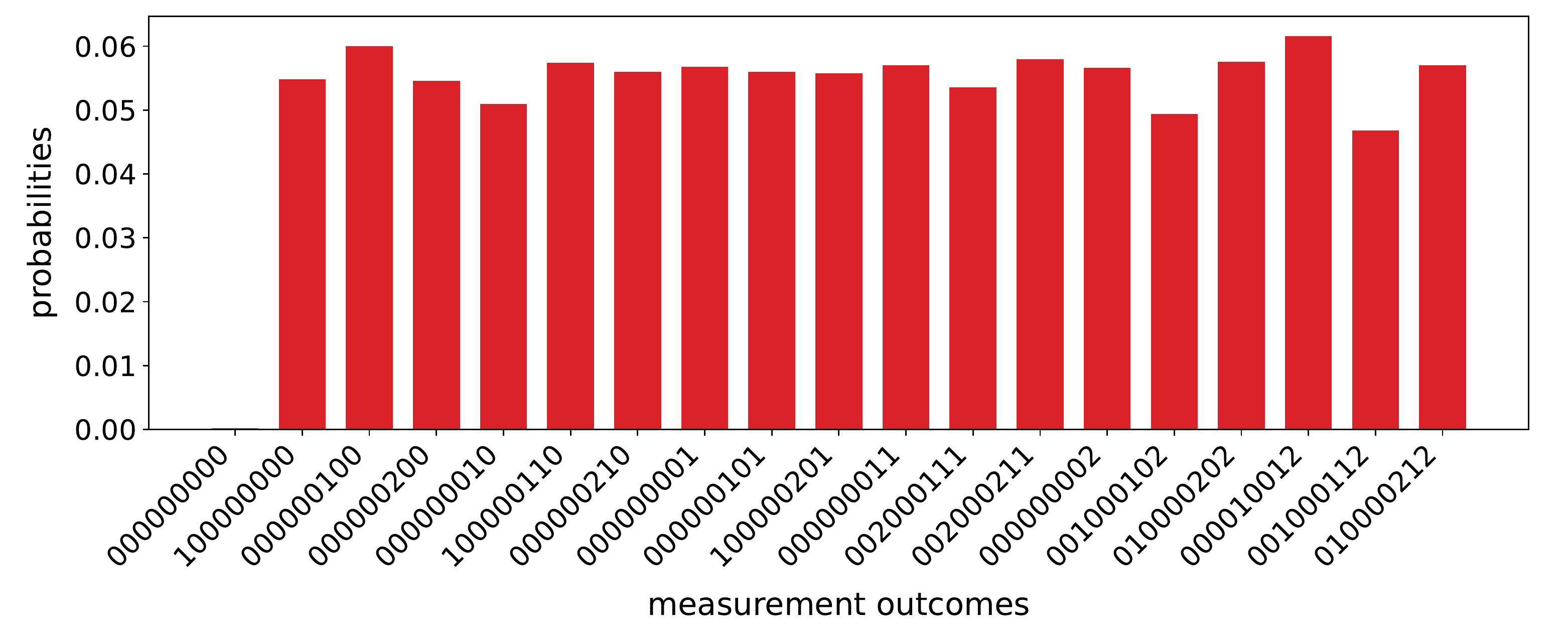}}
    \subfigure[]{\includegraphics[width=0.9\textwidth]{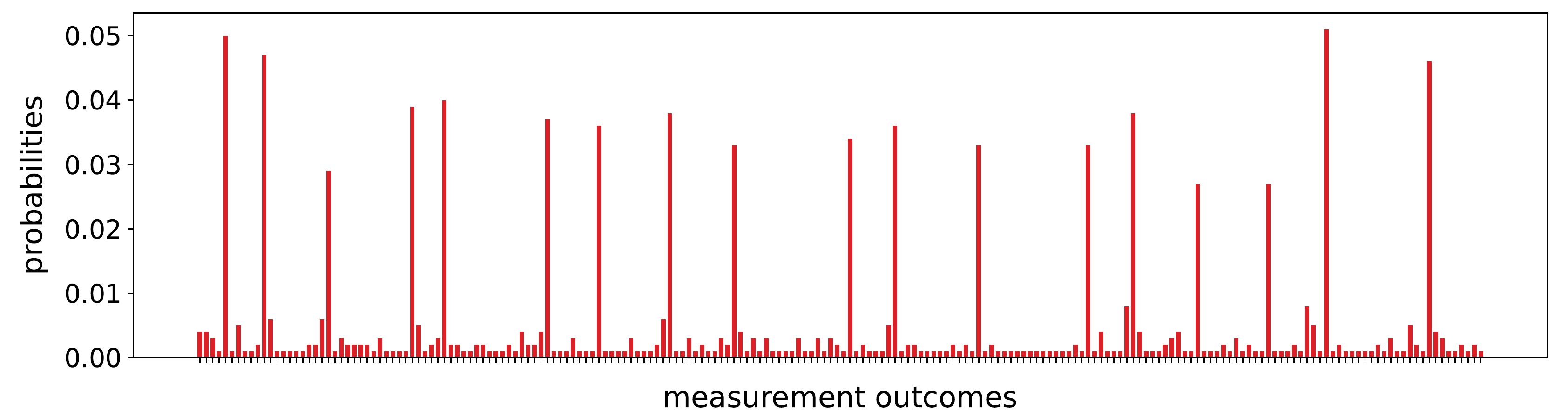}}
    \caption{(a) The circuit for encoding the quantum image state $\vert I\rangle$ in Fig. \ref{color_image}. The register $Q_{1}$ constituted of one qubit $B_{1}$ and one qutrit $T_{1}$ encodes the pixel positions. The qutrit $Q_{2}$ encodes color channels, and the register $Q_{3}$ with six qutrits $t_{i}$ ($i=0, 1, .., 5$) encodes the intensity of the color channels. The six blocks with thin outlines stand for encoding steps for six pixel positions. (b) The probabilities of the outcomes of measurements on all the registers, when the circuit of (a) is simulated using Google Cirq qudit simulator. The data is collected for 5000 instances of measurements. The $x$-axis shows the basis states in the same order as in $\vert I\rangle$ in Fig. \ref{color_image}. (c) The same circuit simulated by decomposing complex gates into single- and two-qudit gates, and adding depolarizing noise to all the gates. The order of magnitude of the noise strength has to be $\leq 10^{-4}$ for a successful image retrieval.}
    \label{image-encoding}
\end{figure*}

\section{Hybrid-qudit representation of color images (HQDQR)}
\label{sec:hybrid_qimage}
In \cite{Dong_2022}, the authors proposed a novel qutrit representation of greyscale images, in which both the pixel positions and pixel values are encoded using qutrits. Compared to 8 qubits in NEQR, it needs 6 qutrits to encode the 256 shades of grey. Since $3^6 > 256 $, a number of the energy levels remain redundant, which can be used for error correction. In this work, we show that a total of 7 qutrits are required to encode the color information of an RGB image. The color channel $\{\mathrm{R,G,B}\}$ can be encoded using the three levels of a qutrit. With this qutrit, a register of 6 qutrits is entangled, which encodes the information about the intensity of each color. Additionally, we consider rectangular images instead of square ones, for which the pixel positions can be encoded using quantum registers constituted of only qubits or qutrits, or both where it is appropriate.

First, we consider the most general case of an $M \times N$ dimensional classical image such that $M=3^{m}$ and $N=2^{n}$. 
The initial quantum state of all the three registers is,
\begin{equation}
    |\Psi_{0}\rangle = \underbrace{|000000\rangle}_\text{6 qutrits} \otimes \underbrace{|0\rangle}_\text{qutrit} \otimes  \underbrace{|000...0\rangle}_\text{m+n qudits} 
\end{equation}
Now, Hadamard operator $H_{3}^{\otimes m} \otimes H_{2}^{\otimes n}$ is applied on the qudits of the last register to transform the state of this register to a fully superposed state.
\begin{eqnarray}
|\Psi_{1}\rangle = \frac{1}{\sqrt{2^{n}3^{m}}}\sum\limits_{Y=0}^{2^{n}}\sum\limits_{X=0}^{3^{m}} |000000\rangle \otimes |0\rangle \otimes |YX\rangle.
\end{eqnarray}
Now, $\mathcal{H}_{3}$ is applied on the second register to prepare the following state,
\begin{equation}
    |\Psi_{2}\rangle = \frac{1}{\sqrt{2^{n}3^{m+1}}}\sum\limits_{Y=0}^{2^{n}}\sum\limits_{X=0}^{3^{m}}|000000\rangle \otimes (|0\rangle + |1\rangle + |2\rangle) \otimes |YX\rangle .
\end{equation}

We assume that the states $|0\rangle$, $\vert 1\rangle$ and $\vert 2\rangle$ represent respectively the color channels R, G and B. In the next step, for each pixel position $|YX\rangle$ and for each color channel, a number of controlled $X$ operations is used to flip the qutrits in the first register, where the control lies on the qudits in the first and second register. We can express this operation as,
\begin{eqnarray}
    &|\Psi_{3}\rangle = \frac{1}{\sqrt{2^{n}3^{m+1}}}\sum\limits_{Y=0}^{3^{m}}\sum\limits_{X=0}^{2^{n}} \Omega_{XY} \big( |000000\rangle \otimes (|0\rangle + |1\rangle \nonumber \\
    &+ |2\rangle)\otimes |YX\rangle \big) \nonumber\\
    &=\frac{1}{\sqrt{2^{n}3^{m+1}}}\sum\limits_{Y=0}^{3^{m}}\sum\limits_{X=0}^{2^{n}} (|R_{XY}\rangle|0\rangle + |G_{XY}\rangle|1\rangle + |B_{XY}\rangle|2\rangle) \nonumber \\
    &\otimes |YX\rangle,
    \label{our_representation}
\end{eqnarray}
where $\Omega_{XY}=\bigotimes\limits_{i=0}^{5}\Omega_{XY}^{i,C}$, $|C_{XY}\rangle = |C^{5}_{XY}C^{4}_{XY}..C^{0}_{XY}\rangle$, $C=\{R, G, B\}$, and each $\Omega_{XY}^{i,C}$ is a higher order generalized hybrid Toffoli gate flipping the $i^{\mathrm{th}}$ qubit of the color channel $C$, in position $|YX\rangle$. Eq. \ref{our_representation} is our desired quantum representation. As an example, we consider a $3\times 2$ dimensional RGB image shown in Fig. \ref{color_image} and its corresponding quantum image state $|I\rangle$.

The different steps for encoding this image are shown in Fig. \ref{image-encoding}(a), where each block surrounded by thin rectangular border stands for one particular pixel position. Thus, by the full use of quantum superposition and entanglement, in place of 24 qubits needed to store the RGB color information, we need only 7 qutrits, when enough auxiliary qutrits are available.

To test the above image encoding and retrieval in a real quantum system, we use Google Cirq's quantum simulator, which provides an architecture of hybrid qudits and corresponding gates. The simulator is a classical simulator which mimics the behaviour of a quantum computer. We build the circuit of Fig. \ref{image-encoding}(a), and finally measure all the registers to retrieve the image state. The result is presented in Fig. \ref{image-encoding}(b). 
The data is collected over 5000 shots of measurements. The probability amplitudes fluctuate around the expected value $\big(\frac{1}{3\sqrt{2}}\big)^{2}\approx 0.577$. This fluctuation may arise because of the random number generator used in the measurement constructor in Cirq. The state where all the registers are in state $\vert 0\rangle$ does not correspond to the original image, though it appears in the measurement outcome with negligible probability.

Next, we test the performance of the image encoding and retrieval under the presence of noise. For this, we decompose each higher order gate into simpler two-qudit gates. All the single- and two-qudit gates in the resulting circuit are subject to depolarizing noise. 
The channels are constructed in Cirq using the Kraus operators for single-qudit and two-qudit depolarizing channel, and they are applied after each gate in the circuit. The measurement outcome of this noisy circuit is presented in Fig. \ref{image-encoding}(c), for which the noise strength corresponding to all the single- and two-qudit channels is $\lambda = 10^{-4}$. The 18 peaks corresponds to the 18 outcomes obtained in Fig. \ref{image-encoding}(b), and the other smaller peaks appear due to the noise. We observe that increasing the single-qudit noise strength by one or two orders of magnitude has little effect on the measurement outcome, while increasing the two-qudit noise strength to $10^{-3}$ can completely randomize the pattern in Fig. \ref{image-encoding}(c), making the image retrieval process highly inefficient. Though the current superconducting qubit based quantum computers typically have single qubit and two-qubit gate errors of the order of $10^{-3}$, tracking the progress of quantum hardware with time implies that in the near future the gate errors will improve to the order of $10^{-4}$ \cite{girvin2014circuit, gokhale_2019}. The image retrieval thus remains feasible in the coming age quantum computers.

\subsection{Special cases}
Above we considered both qubits and qutrits to encode the pixel positions of a rectangular image. However, depending on the dimension of the image, it might be optimum to use only qubits or only qutrits to encode the positions, to minimize the redundancy in the number of energy levels.
In this following subsection, we discuss the quantum representation for RGB images in such cases.

\subsubsection{All-qubit third register}
Suppose the dimension of an image is $M\times N$ such that $M=2^{m}$ and $N=2^{n}$. The quantum image state in this case becomes,
\begin{eqnarray}
    &|I\rangle = \frac{1}{\sqrt{2^{m+n}}}\sum\limits_{Y=0}^{2^{m}}\sum\limits_{X=0}^{2^{n}} (|R_{XY}\rangle|0\rangle + |G_{XY}\rangle|1\rangle \nonumber \\
    &+ |B_{XY}\rangle|2\rangle) \otimes |YX\rangle.
\end{eqnarray}

The generalized higher order Toffoli gates now has only qubits as the control, and a qutrit as the target. It is possible to decompose such gates using auxiliary qutrits as shown in Fig. \ref{hybrid_circuit}, and using effective qutrits as in Fig. \ref{improved_hybrid_circuit}. One can also use auxiliary qubits, the decomposition then can be obtained in terms of Toffoli gates and hybrid Toffoli gates, as shown in Fig. \ref{all_qubit_pixel}. Of course, one can use an extra auxiliary qubit to break the Toffoli gate in Fig. \ref{all_qubit_pixel} into controlled-$X$ gates. We choose to use Toffoli gate since it is already known that a higher order Toffoli gate with $p$ ($p>2$) controls, can be decomposed into $4p-8$ Toffoli gates, when $p-2$ auxiliary qubits are present \cite{yang_2006}.


\begin{figure}
    \centering
    \includegraphics[width=0.45\textwidth]{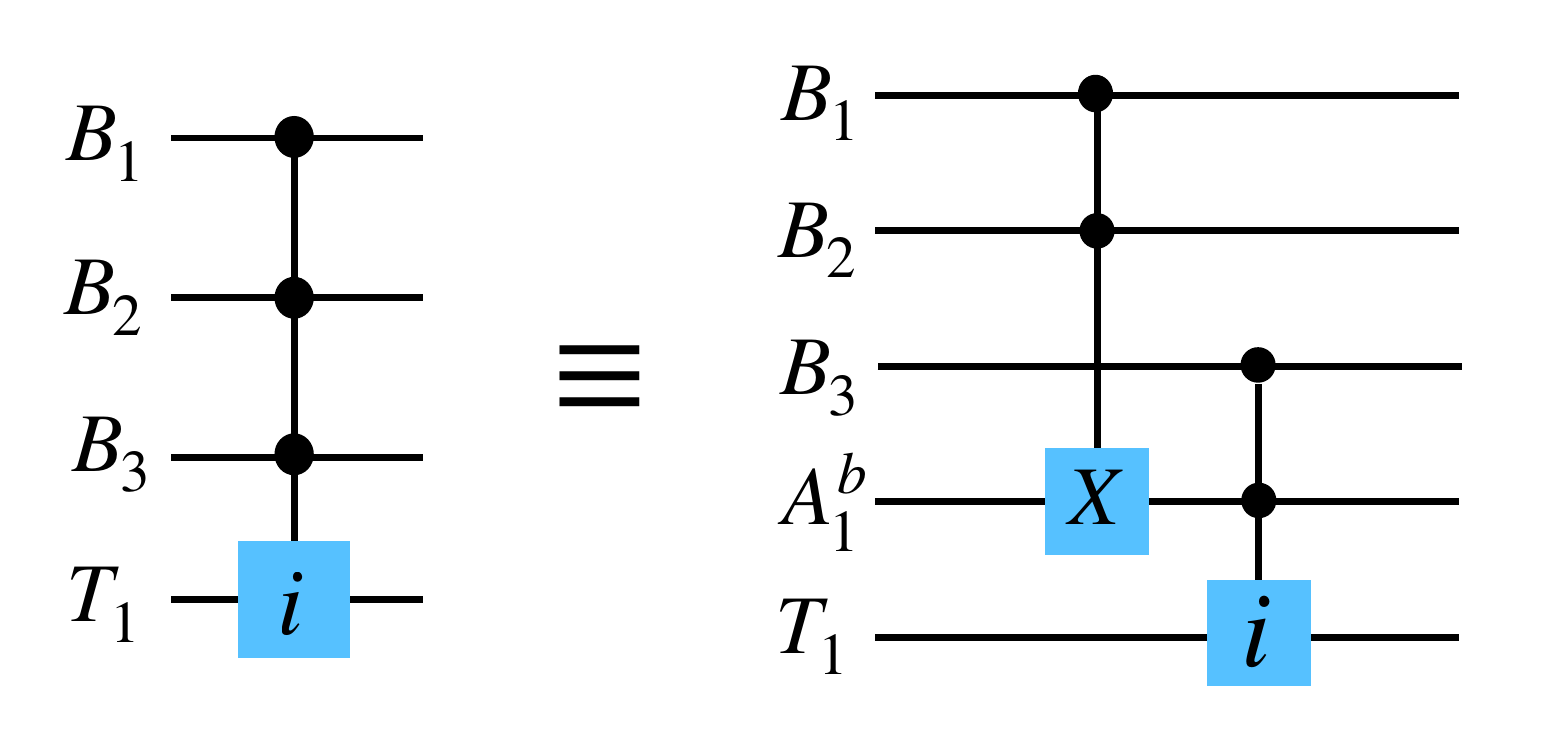}
    
    \caption{Decomposition of a higher order hybrid Toffoli gate where the control qudits $\{B_{1}, B_{2}, B_{3}\}$ are all qubits. The decomposition is shown in terms of auxiliary qubit $A^{b}_{1}$, Toffoli gate, and hybrid Toffoli gate.}
    \label{all_qubit_pixel}
\end{figure}

\subsubsection{All-qutrit third register} If the image dimension is such that $M=3^{m}$ and $N=3^{n}$, the quantum image state becomes,
\begin{eqnarray}
    &|I\rangle = \frac{1}{\sqrt{3^{m+n}}}\sum\limits_{Y=0}^{3^{m}}\sum\limits_{X=0}^{3^{n}} (|R_{XY}\rangle|0\rangle + |G_{XY}\rangle|1\rangle \nonumber \\
    &+ |B_{XY}\rangle|2\rangle) \otimes |YX\rangle.
\end{eqnarray}
The generalized higher order Toffoli gates acts only on qutrits in this case, both as control and target. We do not show the decomposition of these gates, because it can be obtained in the same way as shown in Fig. \ref{hybrid_circuit} and Fig. \ref{general_control}.

\begin{table}
\centering
\begin{tabular}{ c c  }
 \hline
 Representation& Elementary gates used \\
 \hline
 \hline 
 MCQI   & $24\times 2^{4n}-9\times 2^{2n} + 2n + 2$     \\ [0.3em]
 NCQI&   $2n + 24\times 2^{2n}\times 48(n-1) $  \\[0.3em]
 OCQR & $2n+2+24\times 2^{2n} \times 48n$ \\[0.3em]
 HQDQR    &  $m+n+1+18(2^{n}\times 3^{m})(6(m+n)+3)$\\
 \hline
\end{tabular}
\caption{The number of elementary gates (one and two-qudit gates) used for different RGB image encoding methods. For the first three rows the number of pixels in the image is $2^{2n}$, and for the last row the same is $2^{n} \times 3^{m}$.}
\label{table1}
\end{table}

\subsection{Complexity of the encoding method}
Here, we analyze the complexity of the quantum image encoding method in two steps.
\begin{enumerate}
    \item There is total $m+n+1$ single qudit gates applied in the first step to encode the pixel positions and color channels, the complexity of which is $\mathcal{O}(m+n+1)$.
    \item There are $2^{n}\times 3^{m}$ pixels, each pixel having three channels. For each channel, one needs to flip maximum 6 qutrits to encode the intensity. So the maximum number of generalized higher order hybrid Toffoli gates required for each channel is 6. Now, a generalized higher order hybrid Toffoli gate with $m+n+1$ control qudits can be decomposed using maximum of $(6(m+n+1)-3)$ elementary gates. So the total complexity of this step is no more than $(18(2^{n}\times 3^{m})(6(m+n+1)-3)$ which is $\mathcal{O}(N)$, i.e. linear in the number of pixels.
\end{enumerate}
So, the total complexity of the image encoding process is $\mathcal{O}(N)$. It is to be noted that, though the NCQI and OCQR encoding method discussed earlier has similar order of complexity, the number of of elementary gates used for encoding is much less for HQDQR. A comparison between different RGB image representations and the number of elementary gates used has been presented in Table \ref{table1}. For HQDQR, if we assume that the $3^{m}=2^{n}$, the upper bound on the number of elementary gates is $2n + 1 + 18\times 2^{2n} \times (12n +3)$, which is still much less than the number of gates for all of the three other representations in Table \ref{table1}.

\begin{figure*}
    \centering
    \subfigure[]{\includegraphics[width=0.13\textwidth]{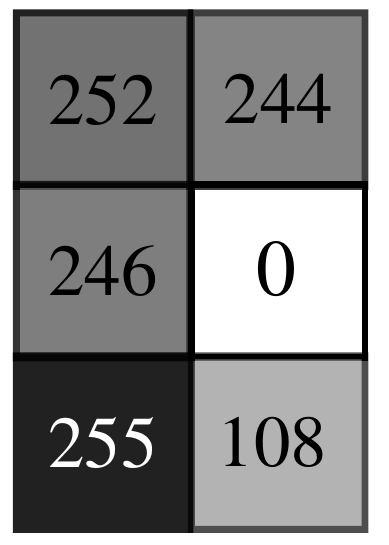}}
    \subfigure[]{\includegraphics[width=0.128\textwidth]{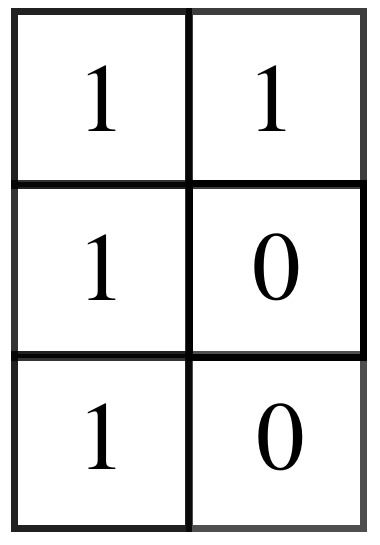}}
    \subfigure[]{\includegraphics[width=0.128\textwidth]{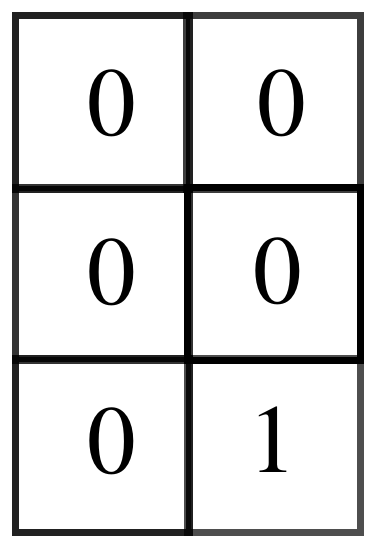}}
    \subfigure[]{\includegraphics[width=0.129\textwidth]{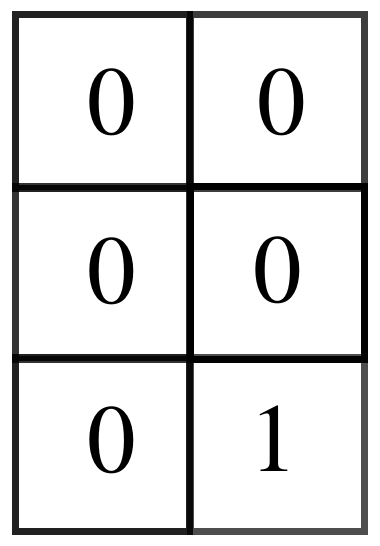}}
    \subfigure[]{\includegraphics[width=0.129\textwidth]{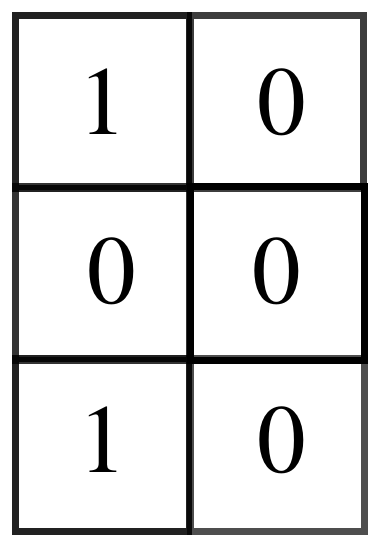}}
    \subfigure[]{\includegraphics[width=0.127\textwidth]{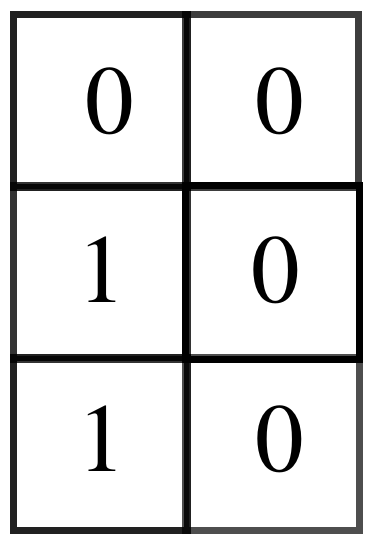}}
    \subfigure[]{\includegraphics[width=0.13\textwidth]{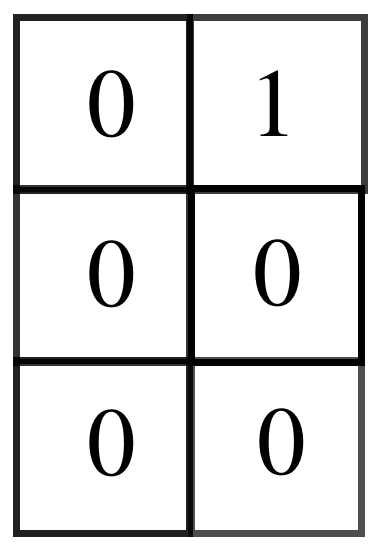}}
    \caption{(a) A $2\times 3$ dimensional greyscale image with pixel values indicated on each pixel. (b)-(c) The six bitplanes of this image.}
    \label{image_compression}
\end{figure*}

\subsection{Image compression}
The complexity of the encoding method, being linear in the number of pixels, becomes significantly high for high-resolution images. It is possible to drastically reduce the numbers of necessary gates for the encoding by using the minimization of logical functions. Here, the logical functions correspond to the bit values of pixel positions. This is a classical procedure, which has been previously discussed in \cite{FRQI_2011, NEQR_2013, Dong_2022} for images encoded using only qubits or qutrits. For our proposed representation, however, both qubits and qutrits are used to encode the pixel positions. Hence, for the best-achieved compression, we must employ an algebra applicable on a combination of binary and ternary logical variables. In the following, we briefly explain the idea of image compression for a simple exemplary image.

To start with, let us consider a $2\times 3$ greyscale image as shown in Fig. \ref{image_compression}(a). In Fig. \ref{image_compression}(b)-(g) we show the 6 bitplanes corresponding to 6 qutrits in the first register, i.e. the $i^{\mathrm{th}}$ bitplane shows the state of the $i^{\mathrm{th}}$ qutrit for all the pixel positions. Let us denote the ternary states $\{0, 1, 2\}$ respectively by $\{x_{+}, x_{0}, x_{-}\}$, and the binary states $\{0, 1\}$ by $\{x, \overline{x}\}$. Now we consider Fig. \ref{image_compression}(b), in which all the pixels along the $1^{\mathrm{st}}$ column has state $|1\rangle$ for the first qutrit. The pixel positions along this particular column are $\vert 00\rangle$, $\vert 10\rangle$ and $\vert 20\rangle$ respectively. It requires 3 generalized Toffoli gates to set the intensities of these three pixels. However, if we consider the sum $S$ of the logical expressions of these pixel positions, and apply the ternary logic algebra \cite{vranesic_1970} as shown below,
\begin{equation}
    S=x_{+}x + x_{0}x + x_{-}x =(x_{+}+x_{0}+x_{-})x =x,
\end{equation}
we see that instead of using three generalized Toffoli gates, we can encode these pixels by using a single controlled-$X$ gate, where the control is only on the position qubit being in state $\vert 0\rangle$. Similarly, in the first row of Fig. \ref{image_compression}(b), the pixels positions are $\vert 00\rangle$ and $\vert 01\rangle$. We use the minimization of the logical expression as following
\begin{equation}
    S^{\prime}=x_{+}x + x_{+}\overline{x}=x_{+}(x+\overline{x})=x_{+},
\end{equation}
so that these two pixels can be encoded using one ternary controlled $X$ gate, where the control is on the position qutrit being in state $\vert 0\rangle$. This idea of image compression for greyscale intensities is translated into RGB image by following the same procedure for each of the color channel intensities. The best achievable compression depends on both the image as well as the encoding method. For example, in \cite{FRQI_2011} the authors achieve $90.63\%$ compression, whereas in \cite{NEQR_2013} the compression obtained was $97.28\%$, and the compression obtained using qutrits in \cite{Dong_2022} was $47.26\%$.

\section{Basic color image processing}
\label{sec:RGB_operations}
In this section, we demonstrate some common RGB image operation using HQDQR encoding.
\subsection{Channel swapping}
The color channel swapping operation $CSO$ is performed to swap the intensities of any two color channels. For our quantum image representation, it can be achieved by applying one of qutrit $X$ gates $\{\sigma^{x}_{01},\sigma^{x}_{12},\sigma^{x}_{02}\}$, whichever applies. For example, to swap the red and green channel, one needs to apply $\sigma^{x}_{01}$ on the second register. 
\begin{eqnarray}
    &CSO_{RG}|I\rangle = \frac{1}{\sqrt{2^{n}3^{m+1}}}\sum\limits_{Y=0}^{3^{m}}\sum\limits_{X=0}^{2^{n}} (|R_{XY}\rangle|1\rangle \nonumber \\
    &+ |G_{XY}\rangle|0\rangle 
    + |B_{XY}\rangle|2\rangle) \otimes |YX\rangle
\end{eqnarray}
For swapping red and blue channel on the other hand, we have to apply $\sigma^{x}_{02}$. The computational complexity of this operation is $\mathcal{O}(1)$.

\subsection{One channel operation} The one channel operation $OCO$ performs a particular transformation to any one of the color channel intensities. For example, the red channel can be transformed by the following,
\begin{eqnarray}
    &OCO_{R}|I\rangle=\frac{1}{\sqrt{2^{n}3^{m+1}}}\sum\limits_{Y=0}^{3^{m}}\sum\limits_{X=0}^{2^{n}} (|R^{\prime}_{XY}\rangle|0\rangle + |G_{XY}\rangle|1\rangle \nonumber \\
    &+ |B_{XY}\rangle|2\rangle) \otimes |YX\rangle,
\end{eqnarray}
where $|R^{\prime}_{XY}\rangle=|R_{XY}^{\prime 5}..R_{XY}^{\prime 0}\rangle$ is the new intensity of the red channel.
This is achieved by using a higher order hybrid Toffoli gate with $m+n+1$ control qudits and a maximum of six target qutrits.

\section{Conclusion}
\label{sec:conclusion}
Quantum image processing is a promising venture towards achieving speed-up in image processing, which in turn is an indispensable task in a plethora of everyday applications. Standing in the era of NISQ devices, it is important to optimize the number of quantum units as well as the depth of a quantum circuit, in order to minimize the effect of noise in the output. In this work, we proposed a quantum representation of RGB images, which uses only 7 qutrits to encode the information about color channels and their intensity. When compared with the existing encoding methods of RGB images, our representation uses least number of quantum units to encode the color information. Moreover, we considered both qubits and qutrits to encode the position information of the pixels, which can be an optimum choice while encoding a general rectangular image while keeping the number of redundant energy levels low. The complexity of our image encoding algorithm is polynomial in the number of pixels. We showed that our representation can be achieved by using much less number of elementary gates compared to the existing encoding methods. The complexity can be further improved by using compression of logical expression corresponding to the pixel positions. Our representation naturally gives rise to a hybrid qubit-qutrit circuit. We demonstrate decomposition of higher order qubit-qutrit gates in terms of simpler single qudit and two-qudit gates in these systems. The current trend of research on higher dimensional quantum units and corresponding gates in these systems, indicates that quantum processing units including qudits will soon become available to users all over the world \cite{cirq_developers}, and thus our work will be a strong candidate for processing of RGB images.

\bibliography{ref}{}

\begin{thebibliography}{44}%
\makeatletter
\providecommand \@ifxundefined [1]{%
 \@ifx{#1\undefined}
}%
\providecommand \@ifnum [1]{%
 \ifnum #1\expandafter \@firstoftwo
 \else \expandafter \@secondoftwo
 \fi
}%
\providecommand \@ifx [1]{%
 \ifx #1\expandafter \@firstoftwo
 \else \expandafter \@secondoftwo
 \fi
}%
\providecommand \natexlab [1]{#1}%
\providecommand \enquote  [1]{``#1''}%
\providecommand \bibnamefont  [1]{#1}%
\providecommand \bibfnamefont [1]{#1}%
\providecommand \citenamefont [1]{#1}%
\providecommand \href@noop [0]{\@secondoftwo}%
\providecommand \href [0]{\begingroup \@sanitize@url \@href}%
\providecommand \@href[1]{\@@startlink{#1}\@@href}%
\providecommand \@@href[1]{\endgroup#1\@@endlink}%
\providecommand \@sanitize@url [0]{\catcode `\\12\catcode `\$12\catcode
  `\&12\catcode `\#12\catcode `\^12\catcode `\_12\catcode `\%12\relax}%
\providecommand \@@startlink[1]{}%
\providecommand \@@endlink[0]{}%
\providecommand \url  [0]{\begingroup\@sanitize@url \@url }%
\providecommand \@url [1]{\endgroup\@href {#1}{\urlprefix }}%
\providecommand \urlprefix  [0]{URL }%
\providecommand \Eprint [0]{\href }%
\providecommand \doibase [0]{https://doi.org/}%
\providecommand \selectlanguage [0]{\@gobble}%
\providecommand \bibinfo  [0]{\@secondoftwo}%
\providecommand \bibfield  [0]{\@secondoftwo}%
\providecommand \translation [1]{[#1]}%
\providecommand \BibitemOpen [0]{}%
\providecommand \bibitemStop [0]{}%
\providecommand \bibitemNoStop [0]{.\EOS\space}%
\providecommand \EOS [0]{\spacefactor3000\relax}%
\providecommand \BibitemShut  [1]{\csname bibitem#1\endcsname}%
\let\auto@bib@innerbib\@empty
\bibitem [{\citenamefont {Muthukrishnan}\ and\ \citenamefont
  {Stroud}(2000)}]{MS_2000}%
  \BibitemOpen
  \bibfield  {author} {\bibinfo {author} {\bibfnamefont {A.}~\bibnamefont
  {Muthukrishnan}}\ and\ \bibinfo {author} {\bibfnamefont {C.~R.}\ \bibnamefont
  {Stroud}},\ }\href {https://doi.org/10.1103/PhysRevA.62.052309} {\bibfield
  {journal} {\bibinfo  {journal} {Phys. Rev. A}\ }\textbf {\bibinfo {volume}
  {62}},\ \bibinfo {pages} {052309} (\bibinfo {year} {2000})}\BibitemShut
  {NoStop}%
\bibitem [{\citenamefont {Leuenberger}\ and\ \citenamefont
  {Loss}(2001)}]{loss_2001}%
  \BibitemOpen
  \bibfield  {author} {\bibinfo {author} {\bibfnamefont {M.~N.}\ \bibnamefont
  {Leuenberger}}\ and\ \bibinfo {author} {\bibfnamefont {D.}~\bibnamefont
  {Loss}},\ }\href {https://doi.org/10.1038/35071024} {\bibfield  {journal}
  {\bibinfo  {journal} {Nature}\ }\textbf {\bibinfo {volume} {410}},\ \bibinfo
  {pages} {789} (\bibinfo {year} {2001})}\BibitemShut {NoStop}%
\bibitem [{\citenamefont {Bartlett}\ \emph {et~al.}(2002)\citenamefont
  {Bartlett}, \citenamefont {de~Guise},\ and\ \citenamefont
  {Sanders}}]{bartlett_2002}%
  \BibitemOpen
  \bibfield  {author} {\bibinfo {author} {\bibfnamefont {S.~D.}\ \bibnamefont
  {Bartlett}}, \bibinfo {author} {\bibfnamefont {H.}~\bibnamefont {de~Guise}},\
  and\ \bibinfo {author} {\bibfnamefont {B.~C.}\ \bibnamefont {Sanders}},\
  }\href {https://doi.org/10.1103/PhysRevA.65.052316} {\bibfield  {journal}
  {\bibinfo  {journal} {Phys. Rev. A}\ }\textbf {\bibinfo {volume} {65}},\
  \bibinfo {pages} {052316} (\bibinfo {year} {2002})}\BibitemShut {NoStop}%
\bibitem [{\citenamefont {Klimov}\ \emph {et~al.}(2003)\citenamefont {Klimov},
  \citenamefont {Guzm\'an}, \citenamefont {Retamal},\ and\ \citenamefont
  {Saavedra}}]{klimov_2003}%
  \BibitemOpen
  \bibfield  {author} {\bibinfo {author} {\bibfnamefont {A.~B.}\ \bibnamefont
  {Klimov}}, \bibinfo {author} {\bibfnamefont {R.}~\bibnamefont {Guzm\'an}},
  \bibinfo {author} {\bibfnamefont {J.~C.}\ \bibnamefont {Retamal}},\ and\
  \bibinfo {author} {\bibfnamefont {C.}~\bibnamefont {Saavedra}},\ }\href
  {https://doi.org/10.1103/PhysRevA.67.062313} {\bibfield  {journal} {\bibinfo
  {journal} {Phys. Rev. A}\ }\textbf {\bibinfo {volume} {67}},\ \bibinfo
  {pages} {062313} (\bibinfo {year} {2003})}\BibitemShut {NoStop}%
\bibitem [{\citenamefont {Ralph}\ \emph {et~al.}(2007)\citenamefont {Ralph},
  \citenamefont {Resch},\ and\ \citenamefont {Gilchrist}}]{ralph_2007}%
  \BibitemOpen
  \bibfield  {author} {\bibinfo {author} {\bibfnamefont {T.~C.}\ \bibnamefont
  {Ralph}}, \bibinfo {author} {\bibfnamefont {K.~J.}\ \bibnamefont {Resch}},\
  and\ \bibinfo {author} {\bibfnamefont {A.}~\bibnamefont {Gilchrist}},\ }\href
  {https://doi.org/10.1103/PhysRevA.75.022313} {\bibfield  {journal} {\bibinfo
  {journal} {Phys. Rev. A}\ }\textbf {\bibinfo {volume} {75}},\ \bibinfo
  {pages} {022313} (\bibinfo {year} {2007})}\BibitemShut {NoStop}%
\bibitem [{\citenamefont {Gottesman}(1998)}]{gottesman_1998}%
  \BibitemOpen
  \bibfield  {author} {\bibinfo {author} {\bibfnamefont {D.}~\bibnamefont
  {Gottesman}},\ }in\ \href@noop {} {\emph {\bibinfo {booktitle} {Selected
  Papers from the First NASA International Conference on Quantum Computing and
  Quantum Communications}}},\ \bibinfo {series and number} {QCQC '98}\
  (\bibinfo  {publisher} {Springer-Verlag},\ \bibinfo {address} {Berlin,
  Heidelberg},\ \bibinfo {year} {1998})\ p.\ \bibinfo {pages}
  {302–313}\BibitemShut {NoStop}%
\bibitem [{\citenamefont {Luo}\ and\ \citenamefont {Wang}(2014)}]{luo_2014}%
  \BibitemOpen
  \bibfield  {author} {\bibinfo {author} {\bibfnamefont {M.}~\bibnamefont
  {Luo}}\ and\ \bibinfo {author} {\bibfnamefont {X.}~\bibnamefont {Wang}},\
  }\href {https://doi.org/10.1007/s11433-014-5551-9} {\bibfield  {journal}
  {\bibinfo  {journal} {Science China Physics, Mechanics \& Astronomy}\
  }\textbf {\bibinfo {volume} {57}},\ \bibinfo {pages} {1712} (\bibinfo {year}
  {2014})}\BibitemShut {NoStop}%
\bibitem [{\citenamefont {Luo}\ \emph {et~al.}(2014)\citenamefont {Luo},
  \citenamefont {Chen}, \citenamefont {Yang},\ and\ \citenamefont
  {Wang}}]{chen_2014}%
  \BibitemOpen
  \bibfield  {author} {\bibinfo {author} {\bibfnamefont {M.-X.}\ \bibnamefont
  {Luo}}, \bibinfo {author} {\bibfnamefont {X.-B.}\ \bibnamefont {Chen}},
  \bibinfo {author} {\bibfnamefont {Y.-X.}\ \bibnamefont {Yang}},\ and\
  \bibinfo {author} {\bibfnamefont {X.}~\bibnamefont {Wang}},\ }\href
  {https://doi.org/10.1038/srep04044} {\bibfield  {journal} {\bibinfo
  {journal} {Scientific Reports}\ }\textbf {\bibinfo {volume} {4}},\ \bibinfo
  {pages} {4044} (\bibinfo {year} {2014})}\BibitemShut {NoStop}%
\bibitem [{\citenamefont {Dogra}\ \emph {et~al.}(2014)\citenamefont {Dogra},
  \citenamefont {Arvind},\ and\ \citenamefont {Dorai}}]{DOGRA20143452}%
  \BibitemOpen
  \bibfield  {author} {\bibinfo {author} {\bibfnamefont {S.}~\bibnamefont
  {Dogra}}, \bibinfo {author} {\bibnamefont {Arvind}},\ and\ \bibinfo {author}
  {\bibfnamefont {K.}~\bibnamefont {Dorai}},\ }\href
  {https://doi.org/https://doi.org/10.1016/j.physleta.2014.10.003} {\bibfield
  {journal} {\bibinfo  {journal} {Physics Letters A}\ }\textbf {\bibinfo
  {volume} {378}},\ \bibinfo {pages} {3452} (\bibinfo {year}
  {2014})}\BibitemShut {NoStop}%
\bibitem [{\citenamefont {Gedik}\ \emph {et~al.}(2015)\citenamefont {Gedik},
  \citenamefont {Silva}, \citenamefont {{\c C}akmak}, \citenamefont {Karpat},
  \citenamefont {Vidoto}, \citenamefont {Soares-Pinto}, \citenamefont
  {deAzevedo},\ and\ \citenamefont {Fanchini}}]{gedik_2015}%
  \BibitemOpen
  \bibfield  {author} {\bibinfo {author} {\bibfnamefont {Z.}~\bibnamefont
  {Gedik}}, \bibinfo {author} {\bibfnamefont {I.~A.}\ \bibnamefont {Silva}},
  \bibinfo {author} {\bibfnamefont {B.}~\bibnamefont {{\c C}akmak}}, \bibinfo
  {author} {\bibfnamefont {G.}~\bibnamefont {Karpat}}, \bibinfo {author}
  {\bibfnamefont {E.~L.~G.}\ \bibnamefont {Vidoto}}, \bibinfo {author}
  {\bibfnamefont {D.~O.}\ \bibnamefont {Soares-Pinto}}, \bibinfo {author}
  {\bibfnamefont {E.~R.}\ \bibnamefont {deAzevedo}},\ and\ \bibinfo {author}
  {\bibfnamefont {F.~F.}\ \bibnamefont {Fanchini}},\ }\href
  {https://doi.org/10.1038/srep14671} {\bibfield  {journal} {\bibinfo
  {journal} {Scientific Reports}\ }\textbf {\bibinfo {volume} {5}},\ \bibinfo
  {pages} {14671} (\bibinfo {year} {2015})}\BibitemShut {NoStop}%
\bibitem [{\citenamefont {Adcock}\ \emph {et~al.}(2016)\citenamefont {Adcock},
  \citenamefont {H\O{}yer},\ and\ \citenamefont {Sanders}}]{adcock_2016}%
  \BibitemOpen
  \bibfield  {author} {\bibinfo {author} {\bibfnamefont {M.~R.}\ \bibnamefont
  {Adcock}}, \bibinfo {author} {\bibfnamefont {P.}~\bibnamefont {H\O{}yer}},\
  and\ \bibinfo {author} {\bibfnamefont {B.~C.}\ \bibnamefont {Sanders}},\
  }\href {https://doi.org/10.1007/s11128-015-1229-0} {\bibfield  {journal}
  {\bibinfo  {journal} {Quantum Information Processing}\ }\textbf {\bibinfo
  {volume} {15}},\ \bibinfo {pages} {1361–1386} (\bibinfo {year}
  {2016})}\BibitemShut {NoStop}%
\bibitem [{\citenamefont {Lu}\ \emph {et~al.}(2019)\citenamefont {Lu},
  \citenamefont {Hu}, \citenamefont {Alshaykh}, \citenamefont {Moore},
  \citenamefont {Wang}, \citenamefont {Imany}, \citenamefont {Weiner},\ and\
  \citenamefont {Kais}}]{Lu2019QuantumPE}%
  \BibitemOpen
  \bibfield  {author} {\bibinfo {author} {\bibfnamefont {H.-H.}\ \bibnamefont
  {Lu}}, \bibinfo {author} {\bibfnamefont {Z.}~\bibnamefont {Hu}}, \bibinfo
  {author} {\bibfnamefont {M.~S.}\ \bibnamefont {Alshaykh}}, \bibinfo {author}
  {\bibfnamefont {A.~J.}\ \bibnamefont {Moore}}, \bibinfo {author}
  {\bibfnamefont {Y.}~\bibnamefont {Wang}}, \bibinfo {author} {\bibfnamefont
  {P.}~\bibnamefont {Imany}}, \bibinfo {author} {\bibfnamefont {A.~M.}\
  \bibnamefont {Weiner}},\ and\ \bibinfo {author} {\bibfnamefont
  {S.}~\bibnamefont {Kais}},\ }\href@noop {} {\bibfield  {journal} {\bibinfo
  {journal} {Advanced Quantum Technologies}\ }\textbf {\bibinfo {volume} {3}}
  (\bibinfo {year} {2019})}\BibitemShut {NoStop}%
\bibitem [{\citenamefont {Gao}\ \emph {et~al.}(2020)\citenamefont {Gao},
  \citenamefont {Erhard}, \citenamefont {Zeilinger},\ and\ \citenamefont
  {Krenn}}]{zeilinger_2020}%
  \BibitemOpen
  \bibfield  {author} {\bibinfo {author} {\bibfnamefont {X.}~\bibnamefont
  {Gao}}, \bibinfo {author} {\bibfnamefont {M.}~\bibnamefont {Erhard}},
  \bibinfo {author} {\bibfnamefont {A.}~\bibnamefont {Zeilinger}},\ and\
  \bibinfo {author} {\bibfnamefont {M.}~\bibnamefont {Krenn}},\ }\href
  {https://doi.org/10.1103/PhysRevLett.125.050501} {\bibfield  {journal}
  {\bibinfo  {journal} {Phys. Rev. Lett.}\ }\textbf {\bibinfo {volume} {125}},\
  \bibinfo {pages} {050501} (\bibinfo {year} {2020})}\BibitemShut {NoStop}%
\bibitem [{\citenamefont {Wang}\ \emph {et~al.}(2020)\citenamefont {Wang},
  \citenamefont {Hu}, \citenamefont {Sanders},\ and\ \citenamefont
  {Kais}}]{wang_2020}%
  \BibitemOpen
  \bibfield  {author} {\bibinfo {author} {\bibfnamefont {Y.}~\bibnamefont
  {Wang}}, \bibinfo {author} {\bibfnamefont {Z.}~\bibnamefont {Hu}}, \bibinfo
  {author} {\bibfnamefont {B.~C.}\ \bibnamefont {Sanders}},\ and\ \bibinfo
  {author} {\bibfnamefont {S.}~\bibnamefont {Kais}},\ }\bibfield  {journal}
  {\bibinfo  {journal} {Frontiers in Physics}\ }\textbf {\bibinfo {volume}
  {8}},\ \href {https://doi.org/10.3389/fphy.2020.589504}
  {10.3389/fphy.2020.589504} (\bibinfo {year} {2020})\BibitemShut {NoStop}%
\bibitem [{\citenamefont {Narvaez}(2021)}]{narvaez_2021}%
  \BibitemOpen
  \bibfield  {author} {\bibinfo {author} {\bibfnamefont {C.~E.~Q.}\
  \bibnamefont {Narvaez}},\ }\href {https://doi.org/10.48550/ARXIV.2109.07282}
  {\bibinfo {title} {Universality for sets of three-valued qubit gates}}
  (\bibinfo {year} {2021})\BibitemShut {NoStop}%
\bibitem [{\citenamefont {Su}\ \emph {et~al.}(2022)\citenamefont {Su},
  \citenamefont {Zhang}, \citenamefont {Bin},\ and\ \citenamefont
  {Yang}}]{su_2022}%
  \BibitemOpen
  \bibfield  {author} {\bibinfo {author} {\bibfnamefont {Q.-P.}\ \bibnamefont
  {Su}}, \bibinfo {author} {\bibfnamefont {Y.}~\bibnamefont {Zhang}}, \bibinfo
  {author} {\bibfnamefont {L.}~\bibnamefont {Bin}},\ and\ \bibinfo {author}
  {\bibfnamefont {C.-P.}\ \bibnamefont {Yang}},\ }\href
  {https://doi.org/10.1103/PhysRevA.105.042434} {\bibfield  {journal} {\bibinfo
   {journal} {Phys. Rev. A}\ }\textbf {\bibinfo {volume} {105}},\ \bibinfo
  {pages} {042434} (\bibinfo {year} {2022})}\BibitemShut {NoStop}%
\bibitem [{\citenamefont {Kerntopf}\ \emph {et~al.}(2004)\citenamefont
  {Kerntopf}, \citenamefont {Perkowski},\ and\ \citenamefont
  {Khan}}]{Kerntopf_2004}%
  \BibitemOpen
  \bibfield  {author} {\bibinfo {author} {\bibfnamefont {P.}~\bibnamefont
  {Kerntopf}}, \bibinfo {author} {\bibfnamefont {M.}~\bibnamefont
  {Perkowski}},\ and\ \bibinfo {author} {\bibfnamefont {M.}~\bibnamefont
  {Khan}},\ }in\ \href {https://doi.org/10.1109/ISMVL.2004.1319922} {\emph
  {\bibinfo {booktitle} {Proceedings. 34th International Symposium on
  Multiple-Valued Logic}}}\ (\bibinfo {year} {2004})\ pp.\ \bibinfo {pages}
  {68--73}\BibitemShut {NoStop}%
\bibitem [{\citenamefont {Bullock}\ \emph {et~al.}(2005)\citenamefont
  {Bullock}, \citenamefont {O'Leary},\ and\ \citenamefont
  {Brennen}}]{bullock_2005}%
  \BibitemOpen
  \bibfield  {author} {\bibinfo {author} {\bibfnamefont {S.~S.}\ \bibnamefont
  {Bullock}}, \bibinfo {author} {\bibfnamefont {D.~P.}\ \bibnamefont
  {O'Leary}},\ and\ \bibinfo {author} {\bibfnamefont {G.~K.}\ \bibnamefont
  {Brennen}},\ }\href {https://doi.org/10.1103/PhysRevLett.94.230502}
  {\bibfield  {journal} {\bibinfo  {journal} {Phys. Rev. Lett.}\ }\textbf
  {\bibinfo {volume} {94}},\ \bibinfo {pages} {230502} (\bibinfo {year}
  {2005})}\BibitemShut {NoStop}%
\bibitem [{\citenamefont {Di}\ and\ \citenamefont {Wei}(2011)}]{Di2011}%
  \BibitemOpen
  \bibfield  {author} {\bibinfo {author} {\bibfnamefont {Y.-M.}\ \bibnamefont
  {Di}}\ and\ \bibinfo {author} {\bibfnamefont {H.}~\bibnamefont {Wei}},\
  }\href@noop {} {\bibfield  {journal} {\bibinfo  {journal} {arXiv: Quantum
  Physics}\ } (\bibinfo {year} {2011})}\BibitemShut {NoStop}%
\bibitem [{\citenamefont {Li}\ \emph {et~al.}(2013)\citenamefont {Li},
  \citenamefont {Gu}, \citenamefont {Liu}, \citenamefont {Lee},\ and\
  \citenamefont {Zhang}}]{zhong_2013}%
  \BibitemOpen
  \bibfield  {author} {\bibinfo {author} {\bibfnamefont {W.-D.}\ \bibnamefont
  {Li}}, \bibinfo {author} {\bibfnamefont {Y.-J.}\ \bibnamefont {Gu}}, \bibinfo
  {author} {\bibfnamefont {K.}~\bibnamefont {Liu}}, \bibinfo {author}
  {\bibfnamefont {Y.-H.}\ \bibnamefont {Lee}},\ and\ \bibinfo {author}
  {\bibfnamefont {Y.-Z.}\ \bibnamefont {Zhang}},\ }\href
  {https://doi.org/10.1103/PhysRevA.88.034303} {\bibfield  {journal} {\bibinfo
  {journal} {Phys. Rev. A}\ }\textbf {\bibinfo {volume} {88}},\ \bibinfo
  {pages} {034303} (\bibinfo {year} {2013})}\BibitemShut {NoStop}%
\bibitem [{\citenamefont {Mischuck}\ and\ \citenamefont
  {M\o{}lmer}(2013)}]{molmer_2013}%
  \BibitemOpen
  \bibfield  {author} {\bibinfo {author} {\bibfnamefont {B.}~\bibnamefont
  {Mischuck}}\ and\ \bibinfo {author} {\bibfnamefont {K.}~\bibnamefont
  {M\o{}lmer}},\ }\href {https://doi.org/10.1103/PhysRevA.87.022341} {\bibfield
   {journal} {\bibinfo  {journal} {Phys. Rev. A}\ }\textbf {\bibinfo {volume}
  {87}},\ \bibinfo {pages} {022341} (\bibinfo {year} {2013})}\BibitemShut
  {NoStop}%
\bibitem [{\citenamefont {Daboul}\ \emph {et~al.}(2003)\citenamefont {Daboul},
  \citenamefont {Wang},\ and\ \citenamefont {Sanders}}]{Daboul_2003}%
  \BibitemOpen
  \bibfield  {author} {\bibinfo {author} {\bibfnamefont {J.}~\bibnamefont
  {Daboul}}, \bibinfo {author} {\bibfnamefont {X.}~\bibnamefont {Wang}},\ and\
  \bibinfo {author} {\bibfnamefont {B.~C.}\ \bibnamefont {Sanders}},\ }\href
  {https://doi.org/10.1088/0305-4470/36/10/312} {\bibfield  {journal} {\bibinfo
   {journal} {Journal of Physics A: Mathematical and General}\ }\textbf
  {\bibinfo {volume} {36}},\ \bibinfo {pages} {2525} (\bibinfo {year}
  {2003})}\BibitemShut {NoStop}%
\bibitem [{\citenamefont {Khan}\ and\ \citenamefont
  {Perkowski}(2006)}]{khan_2006}%
  \BibitemOpen
  \bibfield  {author} {\bibinfo {author} {\bibfnamefont {F.~S.}\ \bibnamefont
  {Khan}}\ and\ \bibinfo {author} {\bibfnamefont {M.}~\bibnamefont
  {Perkowski}},\ }\href {https://doi.org/10.1016/j.tcs.2006.09.006} {\bibfield
  {journal} {\bibinfo  {journal} {Theor. Comput. Sci.}\ }\textbf {\bibinfo
  {volume} {367}},\ \bibinfo {pages} {336–346} (\bibinfo {year}
  {2006})}\BibitemShut {NoStop}%
\bibitem [{\citenamefont {{Dogra}}\ \emph {et~al.}(2015)\citenamefont
  {{Dogra}}, \citenamefont {{Dorai}},\ and\ \citenamefont
  {{Dorai}}}]{dogra_2015}%
  \BibitemOpen
  \bibfield  {author} {\bibinfo {author} {\bibfnamefont {S.}~\bibnamefont
  {{Dogra}}}, \bibinfo {author} {\bibfnamefont {A.}~\bibnamefont {{Dorai}}},\
  and\ \bibinfo {author} {\bibfnamefont {K.}~\bibnamefont {{Dorai}}},\ }\href
  {https://doi.org/10.1142/S0219749915500598} {\bibfield  {journal} {\bibinfo
  {journal} {International Journal of Quantum Information}\ }\textbf {\bibinfo
  {volume} {13}},\ \bibinfo {eid} {1550059-394} (\bibinfo {year} {2015})},\
  \Eprint {https://arxiv.org/abs/1503.06624} {arXiv:1503.06624 [quant-ph]}
  \BibitemShut {NoStop}%
\bibitem [{\citenamefont {Zhang}\ \emph {et~al.}(2015)\citenamefont {Zhang},
  \citenamefont {Lu},\ and\ \citenamefont {Gao}}]{qsobel_2015}%
  \BibitemOpen
  \bibfield  {author} {\bibinfo {author} {\bibfnamefont {Y.}~\bibnamefont
  {Zhang}}, \bibinfo {author} {\bibfnamefont {K.}~\bibnamefont {Lu}},\ and\
  \bibinfo {author} {\bibfnamefont {Y.}~\bibnamefont {Gao}},\ }\href
  {https://doi.org/10.1007/s11432-014-5158-9} {\bibfield  {journal} {\bibinfo
  {journal} {Science China Information Sciences}\ }\textbf {\bibinfo {volume}
  {58}},\ \bibinfo {pages} {1} (\bibinfo {year} {2015})}\BibitemShut {NoStop}%
\bibitem [{\citenamefont {Yao}\ \emph {et~al.}(2017)\citenamefont {Yao},
  \citenamefont {Wang}, \citenamefont {Liao}, \citenamefont {Chen},
  \citenamefont {Pan}, \citenamefont {Li}, \citenamefont {Zhang}, \citenamefont
  {Lin}, \citenamefont {Wang}, \citenamefont {Luo}, \citenamefont {Zheng},
  \citenamefont {Li}, \citenamefont {Zhao}, \citenamefont {Peng},\ and\
  \citenamefont {Suter}}]{suter_2017}%
  \BibitemOpen
  \bibfield  {author} {\bibinfo {author} {\bibfnamefont {X.-W.}\ \bibnamefont
  {Yao}}, \bibinfo {author} {\bibfnamefont {H.}~\bibnamefont {Wang}}, \bibinfo
  {author} {\bibfnamefont {Z.}~\bibnamefont {Liao}}, \bibinfo {author}
  {\bibfnamefont {M.-C.}\ \bibnamefont {Chen}}, \bibinfo {author}
  {\bibfnamefont {J.}~\bibnamefont {Pan}}, \bibinfo {author} {\bibfnamefont
  {J.}~\bibnamefont {Li}}, \bibinfo {author} {\bibfnamefont {K.}~\bibnamefont
  {Zhang}}, \bibinfo {author} {\bibfnamefont {X.}~\bibnamefont {Lin}}, \bibinfo
  {author} {\bibfnamefont {Z.}~\bibnamefont {Wang}}, \bibinfo {author}
  {\bibfnamefont {Z.}~\bibnamefont {Luo}}, \bibinfo {author} {\bibfnamefont
  {W.}~\bibnamefont {Zheng}}, \bibinfo {author} {\bibfnamefont
  {J.}~\bibnamefont {Li}}, \bibinfo {author} {\bibfnamefont {M.}~\bibnamefont
  {Zhao}}, \bibinfo {author} {\bibfnamefont {X.}~\bibnamefont {Peng}},\ and\
  \bibinfo {author} {\bibfnamefont {D.}~\bibnamefont {Suter}},\ }\href
  {https://doi.org/10.1103/PhysRevX.7.031041} {\bibfield  {journal} {\bibinfo
  {journal} {Phys. Rev. X}\ }\textbf {\bibinfo {volume} {7}},\ \bibinfo {pages}
  {031041} (\bibinfo {year} {2017})}\BibitemShut {NoStop}%
\bibitem [{\citenamefont {Le}\ \emph {et~al.}(2011)\citenamefont {Le},
  \citenamefont {Dong},\ and\ \citenamefont {Hirota}}]{FRQI_2011}%
  \BibitemOpen
  \bibfield  {author} {\bibinfo {author} {\bibfnamefont {P.~Q.}\ \bibnamefont
  {Le}}, \bibinfo {author} {\bibfnamefont {F.}~\bibnamefont {Dong}},\ and\
  \bibinfo {author} {\bibfnamefont {K.}~\bibnamefont {Hirota}},\ }\href
  {https://doi.org/10.1007/s11128-010-0177-y} {\bibfield  {journal} {\bibinfo
  {journal} {Quantum Information Processing}\ }\textbf {\bibinfo {volume}
  {10}},\ \bibinfo {pages} {63} (\bibinfo {year} {2011})}\BibitemShut {NoStop}%
\bibitem [{\citenamefont {Sun}\ \emph {et~al.}(2011)\citenamefont {Sun},
  \citenamefont {Le}, \citenamefont {Iliyasu}, \citenamefont {Yan},
  \citenamefont {Garcia}, \citenamefont {Dong},\ and\ \citenamefont
  {Hirota}}]{Sun_2011}%
  \BibitemOpen
  \bibfield  {author} {\bibinfo {author} {\bibfnamefont {B.}~\bibnamefont
  {Sun}}, \bibinfo {author} {\bibfnamefont {P.~Q.}\ \bibnamefont {Le}},
  \bibinfo {author} {\bibfnamefont {A.~M.}\ \bibnamefont {Iliyasu}}, \bibinfo
  {author} {\bibfnamefont {F.}~\bibnamefont {Yan}}, \bibinfo {author}
  {\bibfnamefont {J.~A.}\ \bibnamefont {Garcia}}, \bibinfo {author}
  {\bibfnamefont {F.}~\bibnamefont {Dong}},\ and\ \bibinfo {author}
  {\bibfnamefont {K.}~\bibnamefont {Hirota}},\ }in\ \href
  {https://doi.org/10.1109/WISP.2011.6051718} {\emph {\bibinfo {booktitle}
  {2011 IEEE 7th International Symposium on Intelligent Signal Processing}}}\
  (\bibinfo {year} {2011})\ pp.\ \bibinfo {pages} {1--6}\BibitemShut {NoStop}%
\bibitem [{\citenamefont {Zhang}\ \emph {et~al.}(2013)\citenamefont {Zhang},
  \citenamefont {Lu}, \citenamefont {Gao},\ and\ \citenamefont
  {Wang}}]{NEQR_2013}%
  \BibitemOpen
  \bibfield  {author} {\bibinfo {author} {\bibfnamefont {Y.}~\bibnamefont
  {Zhang}}, \bibinfo {author} {\bibfnamefont {K.}~\bibnamefont {Lu}}, \bibinfo
  {author} {\bibfnamefont {Y.}~\bibnamefont {Gao}},\ and\ \bibinfo {author}
  {\bibfnamefont {M.}~\bibnamefont {Wang}},\ }\href
  {https://doi.org/10.1007/s11128-013-0567-z} {\bibfield  {journal} {\bibinfo
  {journal} {Quantum Information Processing}\ }\textbf {\bibinfo {volume}
  {12}},\ \bibinfo {pages} {2833} (\bibinfo {year} {2013})}\BibitemShut
  {NoStop}%
\bibitem [{\citenamefont {Sang}\ \emph {et~al.}(2017)\citenamefont {Sang},
  \citenamefont {Wang},\ and\ \citenamefont {Li}}]{Sang_2017}%
  \BibitemOpen
  \bibfield  {author} {\bibinfo {author} {\bibfnamefont {J.}~\bibnamefont
  {Sang}}, \bibinfo {author} {\bibfnamefont {S.}~\bibnamefont {Wang}},\ and\
  \bibinfo {author} {\bibfnamefont {Q.}~\bibnamefont {Li}},\ }\href
  {https://doi.org/10.1007/s11128-016-1463-0} {\bibfield  {journal} {\bibinfo
  {journal} {Quantum Information Processing}\ }\textbf {\bibinfo {volume}
  {16}},\ \bibinfo {pages} {1–14} (\bibinfo {year} {2017})}\BibitemShut
  {NoStop}%
\bibitem [{\citenamefont {{Liu}}\ \emph {et~al.}(2018)\citenamefont {{Liu}},
  \citenamefont {{Zhang}}, \citenamefont {{Lu}}, \citenamefont {{Wang}},\ and\
  \citenamefont {{Wang}}}]{OCQR_2018}%
  \BibitemOpen
  \bibfield  {author} {\bibinfo {author} {\bibfnamefont {K.}~\bibnamefont
  {{Liu}}}, \bibinfo {author} {\bibfnamefont {Y.}~\bibnamefont {{Zhang}}},
  \bibinfo {author} {\bibfnamefont {K.}~\bibnamefont {{Lu}}}, \bibinfo {author}
  {\bibfnamefont {X.}~\bibnamefont {{Wang}}},\ and\ \bibinfo {author}
  {\bibfnamefont {X.}~\bibnamefont {{Wang}}},\ }\href
  {https://doi.org/10.1007/s10773-018-3813-4} {\bibfield  {journal} {\bibinfo
  {journal} {International Journal of Theoretical Physics}\ }\textbf {\bibinfo
  {volume} {57}},\ \bibinfo {pages} {2938} (\bibinfo {year}
  {2018})}\BibitemShut {NoStop}%
\bibitem [{\citenamefont {Su}\ \emph {et~al.}(2021)\citenamefont {Su},
  \citenamefont {Guo}, \citenamefont {Liu}, \citenamefont {Lu},\ and\
  \citenamefont {Li}}]{INCQI_2021}%
  \BibitemOpen
  \bibfield  {author} {\bibinfo {author} {\bibfnamefont {J.}~\bibnamefont
  {Su}}, \bibinfo {author} {\bibfnamefont {X.}~\bibnamefont {Guo}}, \bibinfo
  {author} {\bibfnamefont {C.}~\bibnamefont {Liu}}, \bibinfo {author}
  {\bibfnamefont {S.}~\bibnamefont {Lu}},\ and\ \bibinfo {author}
  {\bibfnamefont {L.}~\bibnamefont {Li}},\ }\href
  {https://doi.org/10.1038/s41598-021-93471-7} {\bibfield  {journal} {\bibinfo
  {journal} {Scientific Reports}\ }\textbf {\bibinfo {volume} {11}},\ \bibinfo
  {pages} {13879} (\bibinfo {year} {2021})}\BibitemShut {NoStop}%
\bibitem [{\citenamefont {Dong}\ \emph {et~al.}(2022)\citenamefont {Dong},
  \citenamefont {Lu},\ and\ \citenamefont {Li}}]{Dong_2022}%
  \BibitemOpen
  \bibfield  {author} {\bibinfo {author} {\bibfnamefont {H.}~\bibnamefont
  {Dong}}, \bibinfo {author} {\bibfnamefont {D.}~\bibnamefont {Lu}},\ and\
  \bibinfo {author} {\bibfnamefont {C.}~\bibnamefont {Li}},\ }\href
  {https://doi.org/10.1007/s11128-022-03450-8} {\bibfield  {journal} {\bibinfo
  {journal} {Quant. Inf. Proc.}\ }\textbf {\bibinfo {volume} {21}},\ \bibinfo
  {pages} {108} (\bibinfo {year} {2022})}\BibitemShut {NoStop}%
\bibitem [{\citenamefont {Richards}(2021)}]{richards2021}%
  \BibitemOpen
  \bibfield  {author} {\bibinfo {author} {\bibfnamefont {R.}~\bibnamefont
  {Richards}},\ }\href {https://books.google.it/books?id=HXmwzgEACAAJ} {\emph
  {\bibinfo {title} {Arithmetic Operations in Digital Computers}}}\ (\bibinfo
  {publisher} {Creative Media Partners, LLC},\ \bibinfo {year}
  {2021})\BibitemShut {NoStop}%
\bibitem [{\citenamefont {Developers}(2022)}]{cirq_developers}%
  \BibitemOpen
  \bibfield  {author} {\bibinfo {author} {\bibfnamefont {C.}~\bibnamefont
  {Developers}},\ }\href {https://doi.org/10.5281/zenodo.6599601} {\bibinfo
  {title} {Cirq}} (\bibinfo {year} {2022}),\ \bibinfo {note} {{See full list of
  authors on Github: https://github
  .com/quantumlib/Cirq/graphs/contributors}}\BibitemShut {NoStop}%
\bibitem [{rqv()}]{rqvm2017}%
  \BibitemOpen
  \href@noop {} {\bibinfo {title} {Rigetti computing}},\ \bibinfo
  {howpublished} {\url{https://www.rigetti.com/}}\BibitemShut {NoStop}%
\bibitem [{\citenamefont {Khan}\ \emph {et~al.}(2003)\citenamefont {Khan},
  \citenamefont {Kerntopf},\ and\ \citenamefont {Perkowski}}]{khan_2003}%
  \BibitemOpen
  \bibfield  {author} {\bibinfo {author} {\bibfnamefont {M.~A.}\ \bibnamefont
  {Khan}}, \bibinfo {author} {\bibfnamefont {P.}~\bibnamefont {Kerntopf}},\
  and\ \bibinfo {author} {\bibfnamefont {M.~A.}\ \bibnamefont {Perkowski}},\
  }in\ \href {https://doi.org/10.1109/ISMVL.2003.1201399} {\emph {\bibinfo
  {booktitle} {2013 IEEE 43rd International Symposium on Multiple-Valued
  Logic}}}\ (\bibinfo  {publisher} {IEEE Computer Society},\ \bibinfo {address}
  {Los Alamitos, CA, USA},\ \bibinfo {year} {2003})\ p.\ \bibinfo {pages}
  {146}\BibitemShut {NoStop}%
\bibitem [{\citenamefont {Yurtalan}\ \emph {et~al.}(2020)\citenamefont
  {Yurtalan}, \citenamefont {Shi}, \citenamefont {Kononenko}, \citenamefont
  {Lupascu},\ and\ \citenamefont {Ashhab}}]{yurtalan_2020}%
  \BibitemOpen
  \bibfield  {author} {\bibinfo {author} {\bibfnamefont {M.~A.}\ \bibnamefont
  {Yurtalan}}, \bibinfo {author} {\bibfnamefont {J.}~\bibnamefont {Shi}},
  \bibinfo {author} {\bibfnamefont {M.}~\bibnamefont {Kononenko}}, \bibinfo
  {author} {\bibfnamefont {A.}~\bibnamefont {Lupascu}},\ and\ \bibinfo {author}
  {\bibfnamefont {S.}~\bibnamefont {Ashhab}},\ }\href
  {https://doi.org/10.1103/PhysRevLett.125.180504} {\bibfield  {journal}
  {\bibinfo  {journal} {Phys. Rev. Lett.}\ }\textbf {\bibinfo {volume} {125}},\
  \bibinfo {pages} {180504} (\bibinfo {year} {2020})}\BibitemShut {NoStop}%
\bibitem [{\citenamefont {Goss}\ \emph {et~al.}(2022)\citenamefont {Goss},
  \citenamefont {Morvan}, \citenamefont {Marinelli}, \citenamefont {Mitchell},
  \citenamefont {Nguyen}, \citenamefont {Nail}, \citenamefont {Chen},
  \citenamefont {Jünger}, \citenamefont {Kreikebaum}, \citenamefont
  {Santiago}, \citenamefont {Wallman},\ and\ \citenamefont
  {Siddiqi}}]{Goss_2022}%
  \BibitemOpen
  \bibfield  {author} {\bibinfo {author} {\bibfnamefont {N.}~\bibnamefont
  {Goss}}, \bibinfo {author} {\bibfnamefont {A.}~\bibnamefont {Morvan}},
  \bibinfo {author} {\bibfnamefont {B.}~\bibnamefont {Marinelli}}, \bibinfo
  {author} {\bibfnamefont {B.~K.}\ \bibnamefont {Mitchell}}, \bibinfo {author}
  {\bibfnamefont {L.~B.}\ \bibnamefont {Nguyen}}, \bibinfo {author}
  {\bibfnamefont {R.~K.}\ \bibnamefont {Nail}}, \bibinfo {author}
  {\bibfnamefont {L.}~\bibnamefont {Chen}}, \bibinfo {author} {\bibfnamefont
  {C.}~\bibnamefont {Jünger}}, \bibinfo {author} {\bibfnamefont {J.~M.}\
  \bibnamefont {Kreikebaum}}, \bibinfo {author} {\bibfnamefont {D.~I.}\
  \bibnamefont {Santiago}}, \bibinfo {author} {\bibfnamefont {J.~J.}\
  \bibnamefont {Wallman}},\ and\ \bibinfo {author} {\bibfnamefont
  {I.}~\bibnamefont {Siddiqi}},\ }\href
  {https://doi.org/10.48550/ARXIV.2206.07216} {\bibinfo {title} {High-fidelity
  qutrit entangling gates for superconducting circuits}} (\bibinfo {year}
  {2022})\BibitemShut {NoStop}%
\bibitem [{\citenamefont {Luo}\ \emph {et~al.}(2022)\citenamefont {Luo},
  \citenamefont {Huang}, \citenamefont {Tao}, \citenamefont {Zhang},
  \citenamefont {Zhou}, \citenamefont {Chu}, \citenamefont {Liu}, \citenamefont
  {Wang}, \citenamefont {Cui}, \citenamefont {Liu}, \citenamefont {Yan},
  \citenamefont {Yung}, \citenamefont {Chen}, \citenamefont {Yan},\ and\
  \citenamefont {Yu}}]{luo_2022}%
  \BibitemOpen
  \bibfield  {author} {\bibinfo {author} {\bibfnamefont {K.}~\bibnamefont
  {Luo}}, \bibinfo {author} {\bibfnamefont {W.}~\bibnamefont {Huang}}, \bibinfo
  {author} {\bibfnamefont {Z.}~\bibnamefont {Tao}}, \bibinfo {author}
  {\bibfnamefont {L.}~\bibnamefont {Zhang}}, \bibinfo {author} {\bibfnamefont
  {Y.}~\bibnamefont {Zhou}}, \bibinfo {author} {\bibfnamefont {J.}~\bibnamefont
  {Chu}}, \bibinfo {author} {\bibfnamefont {W.}~\bibnamefont {Liu}}, \bibinfo
  {author} {\bibfnamefont {B.}~\bibnamefont {Wang}}, \bibinfo {author}
  {\bibfnamefont {J.}~\bibnamefont {Cui}}, \bibinfo {author} {\bibfnamefont
  {S.}~\bibnamefont {Liu}}, \bibinfo {author} {\bibfnamefont {F.}~\bibnamefont
  {Yan}}, \bibinfo {author} {\bibfnamefont {M.-H.}\ \bibnamefont {Yung}},
  \bibinfo {author} {\bibfnamefont {Y.}~\bibnamefont {Chen}}, \bibinfo {author}
  {\bibfnamefont {T.}~\bibnamefont {Yan}},\ and\ \bibinfo {author}
  {\bibfnamefont {D.}~\bibnamefont {Yu}},\ }\href
  {https://doi.org/10.48550/ARXIV.2206.11199} {\bibinfo {title} {Experimental
  realization of two qutrits gate with tunable coupling in superconducting
  circuits}} (\bibinfo {year} {2022})\BibitemShut {NoStop}%
\bibitem [{\citenamefont {Gokhale}\ \emph {et~al.}(2019)\citenamefont
  {Gokhale}, \citenamefont {Baker}, \citenamefont {Duckering}, \citenamefont
  {Brown}, \citenamefont {Brown},\ and\ \citenamefont {Chong}}]{gokhale_2019}%
  \BibitemOpen
  \bibfield  {author} {\bibinfo {author} {\bibfnamefont {P.}~\bibnamefont
  {Gokhale}}, \bibinfo {author} {\bibfnamefont {J.~M.}\ \bibnamefont {Baker}},
  \bibinfo {author} {\bibfnamefont {C.}~\bibnamefont {Duckering}}, \bibinfo
  {author} {\bibfnamefont {N.~C.}\ \bibnamefont {Brown}}, \bibinfo {author}
  {\bibfnamefont {K.~R.}\ \bibnamefont {Brown}},\ and\ \bibinfo {author}
  {\bibfnamefont {F.~T.}\ \bibnamefont {Chong}},\ }in\ \href
  {https://doi.org/10.1145/3307650.3322253} {\emph {\bibinfo {booktitle}
  {Proceedings of the 46th International Symposium on Computer
  Architecture}}},\ \bibinfo {series and number} {ISCA '19}\ (\bibinfo
  {publisher} {Association for Computing Machinery},\ \bibinfo {address} {New
  York, NY, USA},\ \bibinfo {year} {2019})\ p.\ \bibinfo {pages}
  {554–566}\BibitemShut {NoStop}%
\bibitem [{\citenamefont {Girvin}(2014)}]{girvin2014circuit}%
  \BibitemOpen
  \bibfield  {author} {\bibinfo {author} {\bibfnamefont {S.~M.}\ \bibnamefont
  {Girvin}},\ }\href@noop {} {\bibfield  {journal} {\bibinfo  {journal}
  {Quantum machines: measurement and control of engineered quantum systems}\ ,\
  \bibinfo {pages} {113}} (\bibinfo {year} {2014})}\BibitemShut {NoStop}%
\bibitem [{\citenamefont {Yang}\ \emph {et~al.}(2006)\citenamefont {Yang},
  \citenamefont {Song}, \citenamefont {Hung}, \citenamefont {Xie},\ and\
  \citenamefont {Perkowski}}]{yang_2006}%
  \BibitemOpen
  \bibfield  {author} {\bibinfo {author} {\bibfnamefont {G.}~\bibnamefont
  {Yang}}, \bibinfo {author} {\bibfnamefont {X.}~\bibnamefont {Song}}, \bibinfo
  {author} {\bibfnamefont {W.~N.~N.}\ \bibnamefont {Hung}}, \bibinfo {author}
  {\bibfnamefont {F.}~\bibnamefont {Xie}},\ and\ \bibinfo {author}
  {\bibfnamefont {M.~A.}\ \bibnamefont {Perkowski}},\ }in\ \href@noop {} {\emph
  {\bibinfo {booktitle} {Theory and Applications of Models of Computation}}},\
  \bibinfo {editor} {edited by\ \bibinfo {editor} {\bibfnamefont {J.-Y.}\
  \bibnamefont {Cai}}, \bibinfo {editor} {\bibfnamefont {S.~B.}\ \bibnamefont
  {Cooper}},\ and\ \bibinfo {editor} {\bibfnamefont {A.}~\bibnamefont {Li}}}\
  (\bibinfo  {publisher} {Springer Berlin Heidelberg},\ \bibinfo {address}
  {Berlin, Heidelberg},\ \bibinfo {year} {2006})\ pp.\ \bibinfo {pages}
  {365--374}\BibitemShut {NoStop}%
\bibitem [{\citenamefont {Vranesic}\ \emph {et~al.}(1970)\citenamefont
  {Vranesic}, \citenamefont {Lee},\ and\ \citenamefont
  {Smith}}]{vranesic_1970}%
  \BibitemOpen
  \bibfield  {author} {\bibinfo {author} {\bibfnamefont {Z.}~\bibnamefont
  {Vranesic}}, \bibinfo {author} {\bibfnamefont {E.}~\bibnamefont {Lee}},\ and\
  \bibinfo {author} {\bibfnamefont {K.}~\bibnamefont {Smith}},\ }\href
  {https://doi.org/10.1109/T-C.1970.222803} {\bibfield  {journal} {\bibinfo
  {journal} {IEEE Transactions on Computers}\ }\textbf {\bibinfo {volume}
  {C-19}},\ \bibinfo {pages} {964} (\bibinfo {year} {1970})}\BibitemShut
  {NoStop}%
\end{thebibliography}%
\bibliographystyle{apsrev4-2}

\end{document}